\documentclass{amsart}

\usepackage{amsmath} 
\usepackage{amssymb}
\usepackage{mathrsfs}

\newtheorem{theorem}{Theorem}[section]

\newtheorem{proposition}[theorem]{Proposition}

\theoremstyle{definition}
\newtheorem{definition}[theorem]{Definition}

\theoremstyle{remark}

\newcount\skewfactor
\def\mathunderaccent#1#2 {\let\theaccent#1\skewfactor#2
\mathpalette\putaccentunder}
\def\putaccentunder#1#2{\oalign{$#1#2$\crcr\hidewidth
\vbox to.2ex{\hbox{$#1\skew\skewfactor\theaccent{}$}\vss}\hidewidth}}

\begin{document}

\title [Logical construction of the ionization energy theory]{Logical construction of the ionization energy theory and the origin of physical categories}
\author {Andrew Das Arulsamy}
\address{Condensed Matter Group, Division of Interdisciplinary Science, F-02-08 Ketumbar Hill, Jalan Ketumbar, 56100 Kuala-Lumpur, Malaysia}
\email{sadwerdna@gmail.com}

\thanks{This work is dedicated to Ambiga Sreenivasan for her unprecedented courage}

\date{\today}

\begin{abstract}
Logical proofs and definitions are developed to establish (1) that the energy-level spacings, $\xi$ for each chemical element (from the periodic table of chemical elements) can be converted to the ionization energies, (2) both $\xi$ and the ionization energies are unique, and (3) the averaged ionization energy of any quantum matter is proportional to the averaged ionization energy of its constituent chemical elements, if and only if $\xi \neq 0$ and $\xi$ is not an irrelevant constant. Physical sets are then constructed such that they are members of a specific physical class where each class belongs to a specific physical category, $\mathscr{P}$. However, there is not a single structure-preserving functor from one energy-level spacing physical category, $^{\xi}\mathscr{P}$ to another $^{\xi}\mathscr{P}'$. Therefore, the existence of many $^{\xi}\mathscr{P}$ implies the existence of different categories of physical systems and quantum matter.
\end{abstract}

\maketitle
{MSC ~2010}: 81-XX
\numberwithin{equation}{section}
\setcounter{section}{0}

\section{Introduction}

In non-relativistic classical mechanics, one employs the Newton's second law such that $F(\textbf{r},t) = m{\rm d}^2\textbf{r}(t)/{\rm d}t^2$, which specifies the trajectory of a particle (mass, $m$), which can be used to determine both the position and momentum of this particle at any given time. In this case, we just need to know $\textbf{r}(t)$, which can be obtained from $F(\textbf{r},t) = -\sum_i\partial V_i(\textbf{r},t)/\partial\textbf{r}(t)$ with an appropriate initial condition, $\textbf{r}(t = 0)$ where each $i$ represents a type of potential acting on the particle. The particle's momentum, $m\dot{\textbf{r}}(t)$ can be obtained from its kinetic energy, $(1/2)m\dot{\textbf{r}}(t)^2$, which means, the particle has a precise momentum and a precise position at any given time $t$. This also means that, a classical particle, which is by definition, stable and has a rigid structure and shape, can exist with a well-defined position and a precise momentum for all time, even if it is unbounded in an empty and free space (with $\sum_iV_i(\textbf{r}) = 0$). Simply put, classical particles are literally ``dead''.

On the other hand, in non-relativistic quantum mechanics, one requires the wave function ($\Psi(\textbf{r},t)$), instead of $\textbf{r}(t)$, to predict any physical quantity associated to this quantum system. In fact, in any quantum system, a well-defined $\textbf{r}(t)$ for a quantum particle (regardless whether this particle is confined within the quantum system or not) does not exist, regardless whether $\sum_iV_i(\textbf{r}) = 0$ (free quantum particle) or not (bounded and/or scattered quantum particle). In fact, the term ``particle'' is used here entirely for convenience, where it should never be regarded as a classical particle. 

For example, quantum theory puts an end to a classical notion that reads, A is made of B, B is made of C, $\cdots$, V is made of Werewolves, and it is Werewolves all the way down. On the contrary, the quantum notion reads, an electron (one of the quantum particles in atoms) has a quantized energy with additional properties attached to this ``energy'', such as electron charge, $-e$, mass, $m$ and spin, $s$, and these quantities do vary quantitatively, depending on the environment, due to screening and interactions. Due to the above quantities ($-e$, $m$ and $s$), we can think of the quantized energy as ``particle-like'', and this is the reason why an electron is sometimes regarded as both particle-like (due to $-e$, $m$ and $s$) and wave-like (due to energy) at the same time, giving rise to the popular notion of particle-wave duality. But this does not imply, in any way, that one can consider an electron as classical particle-like. Therefore, the above-mentioned particle-wave duality (strictly) does not allow the existence of a wave-particle ``thingy'', something like a wave-guiding particle, or a particle-guiding wave. The reason is that such a ``thingy'' will only lead us to ask--- what is the constituent of this wave-guiding particle (or the wave-guided particle)? Obviously, this latter interpretation leads us back to square one (to the classical notion described above).   

Anyway, the proper and real wave function, $\Psi(\textbf{r},t)$ can only be obtained by solving the Schr$\ddot{\rm o}$dinger equation~\cite{sch,sch2,sch3},  
\begin {eqnarray}
i\hbar\frac{\partial \Psi(\textbf{r},t)}{\partial t} = \bigg[-\frac{\hbar^2}{2m}\nabla^2 + V\bigg]\Psi(\textbf{r},t) = H\Psi(\textbf{r},t), ~~~\nabla^2 = \frac{\partial^2}{\partial \textbf{r}^2}. \label{eq:IN1}
\end {eqnarray}  
Here, $H$ and $\textbf{r}$ are the Hamilton and position operators, respectively, $\hbar = h/2\pi$ where $h$ is the Planck constant~\cite{planck}, $m$ denotes the electron mass. This wave function is postulated such that $\Psi(\textbf{r},t)$ completely defines the dynamical state of a given system subject to the Heisenberg uncertainty principle~\cite{heisen}, contrary to the Einstein-Podolski-Rosen arguments that require simultaneous reality~\cite{epr}. The above uncertainty principle gives rise to a probability law in accordance with Born's statistical interpretation of the wave function~\cite{born,born2}. In particular, $\Psi(\textbf{r},t)$ by definition has a certain spatial extension such that one cannot attribute a precise position or a momentum to a quantum particle at time $t$. This means that we can only determine the probability of finding the particle's position or momentum ($\wp$) in a given region of space at time $t$. This probability is given by~\cite{born,born2}     
\begin {eqnarray}
&&\textsc{Prob} = \int_{-\infty}^{\infty}\Psi(\textbf{r},t)^*\Psi(\textbf{r},t){\rm d}\textbf{r} = \int_{-\infty}^{\infty}|\Psi(\textbf{r},t)|^2{\rm d}\textbf{r} = \langle\Psi(\textbf{r},t)|\Psi(\textbf{r},t)\rangle, \label{eq:IN2}
\end {eqnarray}
where $r = \langle\Psi(\textbf{r},t)|\textbf{r}|\Psi(\textbf{r},t)\rangle$, $\langle\wp\rangle = \langle\Psi(\textbf{r},t)|\wp|\Psi(\textbf{r},t)\rangle$, $\Psi(\textbf{r},t)$ is orthonormalized such that $\langle\Psi(\textbf{r},t)|\Psi(\textbf{r},t)\rangle = 1$, $\Psi(\textbf{r},t) = \psi_i(\textbf{r},t)\psi_j(\textbf{r},t)$, $\langle\psi_i(\textbf{r},t)|\psi_j(\textbf{r},t)\rangle = \delta_{ij}$ where $\delta_{ij} = 1$ if $i = j$ and $\delta_{ij} = 0$ if $i \neq j$, and $[\textbf{r},\wp] = \textbf{r}\wp - \wp\textbf{r} = i\hbar$. 

But it turns out, Eq.~(\ref{eq:IN1}) cannot be solved analytically to obtain the real $\Psi(\textbf{r})$ (time-independent) for atoms other than atomic hydrogen if 
\begin {eqnarray}
V = \sum_{i,I}\frac{(-e)(+e)}{4\pi\epsilon_0|\textbf{r}_{i} - \textbf{R}_I|} + \sum_{i \neq j}\frac{(-e)(-e)}{4\pi\epsilon_0|\textbf{r}_{i} - \textbf{r}_j|}, \label{eq:IN2a}
\end {eqnarray}
in which, the first and second terms on the right-hand side of Eq.~(\ref{eq:IN2a}) are the Coulomb potentials due to electron($\textbf{r}_i$)-ion($\textbf{R}_I$) and electron($\textbf{r}_i$)-electron($\textbf{r}_j$) interactions, respectively~\cite{bethe}, where $e$ is electron charge and $\epsilon_0$ is the permittivity of free space. Here, the nucleus coordinate for atoms can be positioned at the origin where $\textbf{R}_I = \textbf{R} = 0$. 

\subsection{Hartree-Fock theory}

In response to Eq.~(\ref{eq:IN2a}), all practical applications in natural sciences requires one resorting to finding the unique but ``unreal'' $\Psi(\textbf{r})$ by other means, namely, via the linear combination of atomic orbitals, molecular orbitals, valence bonds, Hartree-Fock-Slater wave functions, and their improved variants in the chemical sciences~\cite{levin,par,pople}. Whereas, for physical sciences and solids, we have the density functional theory, which is based on the Hohenberg-Kohn-Sham variational principle that relies on the existence of a one-to-one correspondence between the potential ($V^{\rm HKS}(\textbf{r})$) and the total electron density~\cite{kohn,kohn2,kohn3,pople2,kaxi}. Here, ``unreal'' $\Psi(\textbf{r})$ simply means that $\Psi(\textbf{r})$ is not obtained by solving a given Hamiltonian, or $\Psi(\textbf{r})$ is a general solution obtained or generated by some educated guesses. Anyway, to see the importance of this Hohenberg-Kohn-Sham principle, we need to recall the Hartree-Fock total energy ($E^{\rm HF}$)~\cite{kaxi}. After invoking the Born-Oppenheimer approximation~\cite{bo}    
\begin {eqnarray}
E^{\rm HF} &=& \langle\Psi^{\rm HF}(\textbf{r})|H|\Psi^{\rm HF}(\textbf{r})\rangle \nonumber \\&=&
\int \psi^{\rm HF}_i(\textbf{r}_i)^*\bigg[\sum_i\frac{-\hbar^2}{2m}\nabla_i^2 - \frac{e^2}{4\pi\epsilon_0}\bigg(\sum_{i,I}\frac{Z_I}{|\textbf{R}_I-\textbf{r}_i|} - \frac{1}{2}\sum_{I \neq J}\frac{Z_IZ_J}{|\textbf{R}_I-\textbf{R}_J|}\bigg)\bigg]\psi^{\rm HF}_i(\textbf{r}_i){\rm d}\textbf{r}_i \nonumber \\&& + \frac{1}{2}\sum_{i \neq j}\int\psi_j^{\rm HF}(\textbf{r}_j)^*\psi_i^{\rm HF}(\textbf{r}_i)^*\frac{e^2}{4\pi\epsilon_0|\textbf{r}_i-\textbf{r}_j|}\psi_i^{\rm HF}(\textbf{r}_i)\psi_j^{\rm HF}(\textbf{r}_j){\rm d}\textbf{r}_i{\rm d}\textbf{r}_j \nonumber \\&& + \frac{1}{2}\sum_{i \neq j}\int\psi_j^{\rm HF}(\textbf{r}_j)^*\psi_i^{\rm HF}(\textbf{r}_i)^*\frac{e^2}{4\pi\epsilon_0|\textbf{r}_i-\textbf{r}_j|}\big[-\psi_i^{\rm HF}(\textbf{r}_j)\psi_j^{\rm HF}(\textbf{r}_i)\big]{\rm d}\textbf{r}_i{\rm d}\textbf{r}_j, \nonumber \\&& \label{eq:IN3} 
\end {eqnarray}
where~\cite{sla,sla2,sla3}
\begin {eqnarray}
\Psi^{\rm HF}(\textbf{r}_1, \textbf{r}_2, \cdots, \textbf{r}_n) &=& \big[\psi^{\rm HF}_1(\textbf{r}_1)\psi^{\rm HF}_2(\textbf{r}_2) \cdots \psi^{\rm HF}_n(\textbf{r}_n)\big] \nonumber \\&& -\big[\psi^{\rm HF}_1(\textbf{r}_2)\psi^{\rm HF}_2(\textbf{r}_1) \cdots \psi^{\rm HF}_n(\textbf{r}_n)\big] + \cdots,  \label{eq:IN4} 
\end {eqnarray}
and the negative signs in the last term on the right-hand side of Eq.~(\ref{eq:IN3}) and in Eq.~(\ref{eq:IN4}) are due to the Pauli exclusion principle~\cite{pauli} 
\begin {eqnarray}
\Psi^{\rm HF}(\textbf{r}_1,\textbf{r}_2) = -\Psi^{\rm HF}(\textbf{r}_2,\textbf{r}_1). \label{eq:IN5} 
\end {eqnarray}
Obviously, Eq.~(\ref{eq:IN3}) cannot be solved for any real system to obtain $\Psi^{\rm HF}(\textbf{r})$, however one can guess the general solution, $\Psi^{\rm HF}(\textbf{r}_i)$ that solves Eq.~(\ref{eq:IN3}) by iteration (numerically). Real systems here mean atoms (other than atomic hydrogen), molecules and solids. For example, Eq.~(\ref{eq:IN3}) can be solved exactly if $\Psi^{\rm HF}(\textbf{r}_i)$ is taken to be the set of plane waves (that represent free electrons), and if $\frac{1}{2}\sum_{I \neq J}\frac{Z_IZ_Je^2}{4\pi\epsilon_0|\textbf{R}_I-\textbf{R}_J|}$ is approximated as a constant independent of the electronic wave functions~\cite{ps} due to the Born-Oppenheimer approximation~\cite{bo}. Equation~(\ref{eq:IN3}) becomes computationally expensive (because large number of iterations are required) with increasing number of atoms or electrons due to Eq.~(\ref{eq:IN4}). 

\subsection{Density functional theory}

As a consequence of large number of required iterations, further simplifications are required to even obtain a proper numerical solution ($\Psi^{\rm HF}(\textbf{r})$) to Eq.~(\ref{eq:IN3}) for solids. One such simplified approximation is provided by the density functional theory (DFT)~\cite{kohn,kohn2,kohn3,pople2}. Contrary to the Hartree-Fock theory (Eq.~(\ref{eq:IN3})) that needs a proper general solution, $\Psi^{\rm HF}(\textbf{r})$ with respect to Eq.~(\ref{eq:IN4}) (otherwise, the total energy may not approach the real value or converge), DFT just requires any general solution, namely, $\Psi^{\rm DFT}(\textbf{r})$. Meaning, with less microscopic details incorporated into $\Psi^{\rm DFT}(\textbf{r})$. However, $\Psi^{\rm DFT}(\textbf{r})$ needs to be iterated in order to obtain the correct electron density such that~\cite{kohn,kohn2,kohn3,pople2}   
\begin {eqnarray}
\langle\Psi^{\rm DFT}(\textbf{r})|V^{\rm external}(\textbf{r})|\Psi^{\rm DFT}(\textbf{r})\rangle &=& \sum_i\int\psi^{\rm DFT}(\textbf{r}_i)^*V^{\rm external}(\textbf{r}_i)\psi^{\rm DFT}(\textbf{r}_i) {\rm d}\textbf{r}_i \nonumber \\&=& \int n(\textbf{r})V^{\rm external}(\textbf{r}){\rm d}\textbf{r}, \label{eq:IN6} 
\end {eqnarray}
where $n(\textbf{r})$ is the ground state electron density. The Hohenberg-Kohn-Sham theorem states that each unique electron density corresponds to an external potential such that $n(\textbf{r}) \Rightarrow V^{\rm external}(\textbf{r})$ and $n(\textbf{r})' \Rightarrow V^{\rm external}(\textbf{r})'$ where $n(\textbf{r}) \neq n(\textbf{r})'$, $V^{\rm external}(\textbf{r}) \neq V^{\rm external}(\textbf{r})'$ and $V^{\rm external}(\textbf{r})$ is different from $V^{\rm external}(\textbf{r})'$ in a nontrivial way. For example,~\cite{kohn,kohn2,kohn3,pople2}  
\begin {eqnarray}
V^{\rm external}(\textbf{r}) = -\sum_{i,I}\frac{Z_Ie^2}{4\pi\epsilon_0|\textbf{R}_I-\textbf{r}_i|}, \label{eq:IN7} 
\end {eqnarray}
which is a sufficient representative. This implies $V^{\rm external}(\textbf{r})$ is not uniquely defined, which means, it does not have to be defined by Eq.~(\ref{eq:IN7}), instead it can be any other proper potential function. 

\texttt{NOTE 1}: Even though DFT does not deal with a unique potential function for a given unique true electron density, for example, DFT deals only with an arbitrary $V^{\rm external}(\textbf{r})$, which exclusively gives a unique $n(\textbf{r})$. In contrast, IET logically proves the existence of a unique potential and Hamiltonian due to the unique numbers of electrons ($\zeta$), protons ($\tau$) and neutrons ($\eta$) in a quantum system via the unique function, $f(\zeta,\eta,\tau)$. The uniqueness of $f(\zeta,\eta,\tau)$ implies the ionization energies (energy levels) and the energy-level spacings associated to an atom or a molecule or a solid are also unique.  

Anyway, the DFT ground state energy based on the Hohenberg-Kohn variational theorem, after invoking the Born-Oppenheimer approximation, can be obtained from~\cite{kohn,kohn2,kohn3,pople2} 
\begin {eqnarray}
E[n_0] \leq E^{\rm DFT}[n] &=& \int n(\textbf{r})V^{\rm external}(\textbf{r}){\rm d}\textbf{r} + \sum_i\bigg\langle\psi^{\rm DFT}(\textbf{r}_i)\bigg|\frac{-\hbar^2}{2m}\nabla_i^2\bigg|\psi^{\rm DFT}(\textbf{r}_i)\bigg\rangle \nonumber \\&& + \sum_{i \neq j}\bigg\langle\psi^{\rm DFT}(\textbf{r}_j)\psi^{\rm DFT}(\textbf{r}_i)\bigg|\frac{e^2}{4\pi\epsilon_0|\textbf{r}_i-\textbf{r}_j|}\bigg|\psi^{\rm DFT}(\textbf{r}_i)\psi^{\rm DFT}(\textbf{r}_j)\bigg\rangle \nonumber \\&=& \int n(\textbf{r})V^{\rm external}(\textbf{r}){\rm d}\textbf{r} + \langle T\rangle[n] + \langle V_{\rm ee}\rangle[n], \label{eq:IN8} 
\end {eqnarray}
while the true ground state energy functional, by definition, is given by~\cite{kohn,kohn2,kohn3,pople2}
\begin {eqnarray}
&&E[n_0] = \int n_0(\textbf{r})V^{\rm external}(\textbf{r}){\rm d}\textbf{r} + \langle T\rangle[n_0] + \langle V_{\rm ee}\rangle[n_0], \label{eq:IN9} 
\end {eqnarray}
where $n_0$ is the true ground state electron density, and Eq.~(\ref{eq:IN9}) is exact in principle. 

\texttt{NOTE 2}: From now on, we no longer discuss how $\psi^{\rm DFT}(\textbf{r})$ is determined or approximated (in terms of Kohn-Sham orbitals) because our focus here is to find any association and/or differences between DFT and the ionization energy theory (IET), independent of the wave function. The reason is that the IET-construction and proofs (developed in the subsequent sections) do not require any explicit knowledge on any real or arbitrary wave functions. Of course, for every quantum system, there exist a corresponding real (unique and true) wave function and a real Hamiltonian (see \texttt{NOTE} 1).

Equations~(\ref{eq:IN8}) and~(\ref{eq:IN9}) are both exact in principle, however, the mathematical structure of these functionals, $\langle T\rangle[n]$, $\langle V_{\rm ee}\rangle[n]$, $\langle T\rangle[n_0]$ and $\langle V_{\rm ee}\rangle[n_0]$ are unknown. Therefore, one needs to invoke the Kohn-Sham method to rewrite Eq.~(\ref{eq:IN8}), where one first define a non-interacting reference Hamiltonian~\cite{kohn,kohn2,kohn3,pople2}  
\begin {eqnarray}
&&H^{\rm KS}_{\rm reference} = -\sum_i\frac{\hbar^2}{2m}\nabla_i^2 - \frac{e^2}{4\pi\epsilon_0}\bigg(\sum_{i,I}\frac{Z_I}{|\textbf{R}_I-\textbf{r}_i|}\bigg), \label{eq:IN10} 
\end {eqnarray}
and the above-mentioned unknown functional (given in Eqs.~(\ref{eq:IN8}) and~(\ref{eq:IN9})) are defined as follows~\cite{kohn,kohn2,kohn3,pople2}
\begin {eqnarray}
&&\Delta \langle T\rangle[n] = \langle T\rangle[n] - \langle T_{\rm reference}^\textbf{\rm KS}\rangle[n], \label{eq:IN11} 
\end {eqnarray}
where $\langle T\rangle[n]$ is the average ground state electronic kinetic energy functional, $\langle T_{\rm reference}^\textbf{\rm KS}\rangle[n]$ can be obtained from Eq.~(\ref{eq:IN10}), and
\begin {eqnarray}
&&\Delta \langle V_{\rm ee}\rangle[n] = \langle V_{\rm ee}\rangle[n] - \frac{1}{2}\int \frac{n(\textbf{r}_i)n(\textbf{r}_j)}{|\textbf{r}_i - \textbf{r}_j|}{\rm d}\textbf{r}_i{\rm d}\textbf{r}_j.\label{eq:IN12} 
\end {eqnarray}
Here $\Delta \neq \nabla^2$ and $\langle V_{\rm ee}\rangle[n]$ is the average ground state electron-electron interaction potential energy functional. The second term on the right-hand side of Eq.~(\ref{eq:IN12}) records the changes to the electrostatic repulsion between electrons. Substituting Eqs.~(\ref{eq:IN11}) and~(\ref{eq:IN12}) into Eq.~(\ref{eq:IN8}) leads us to~\cite{levin,par} 
\begin {eqnarray}
E[n] &=& \int n(\textbf{r})V^{\rm external}(\textbf{r}){\rm d}\textbf{r} + \langle T^{\rm KS}_{\rm reference}\rangle[n] + \frac{1}{2}\int \frac{n(\textbf{r}_i)n(\textbf{r}_j)}{|\textbf{r}_i - \textbf{r}_j|}{\rm d}\textbf{r}_i{\rm d}\textbf{r}_j \nonumber \\&& + \Delta\langle V_{\rm ee}\rangle[n] + \Delta\langle T\rangle[n], \label{eq:IN13} 
\end {eqnarray}
in which, $\Delta\langle V_{\rm ee}\rangle[n] + \Delta\langle T\rangle[n] = E_{\rm xc}[n] = E_{\rm x}[n] + E_{\rm c}[n]$ remains unknown, where $E_{\rm xc}[n]$ is defined as the exchange(x)-correlation(c) energy functional~\cite{levin,par}. There are several approximations available to calculate $E_{\rm xc}[n]$, namely, the local density ($n(\textbf{r})$) and local spin density ($n^{\uparrow}(\textbf{r})$, $n^{\downarrow}(\textbf{r})$) approximations (LDA and LSDA), the $X\alpha$ method that completely ignores the correlation-energy functional by assuming $E_{\rm c} \ll E_{\rm x}$, where $E_{\rm x} \approx E^{\rm LDA}_{\rm x}$. Here, the spin, $s$ denotes spin-up $\uparrow$ or -down $\downarrow$. The Hartree-Fock functional DFT uses Kohn-Sham orbitals instead of $\Psi^{\rm HF}(\textbf{r})$ (given in Eq.~(\ref{eq:IN3})) to obtain $E_{\rm x}$ (given by the last term on the right-hand side of Eq.~(\ref{eq:IN3})). In this approach, $E_{\rm c}$ is calculated from LDA or LSDA method. Finally, we also have the gradient-corrected hybrid functionals that consider $n^{\uparrow}(\textbf{r})$ and $n^{\downarrow}(\textbf{r})$, as well as the gradient-corrected electron densities ($\nabla n^{\uparrow}(\textbf{r})$ and $\nabla n^{\downarrow}(\textbf{r})$)~\cite{levin,par}. This gradient correction is also known as the generalized gradient approximation (GGA). It is strange to call the above ``approximations'' as approximations because the exact $E_{\rm xc}$ is unknown. For example, one needs to know the exact exchange-correlation functional in order to derive the approximated functional. Hence, the above ``approximations'' seem to be educated guesses, rather than approximations. 

But never mind, physically, $E_{\rm xc}[n]$ exists due to two effects--- (i) electron exchange that gives rise to exchange energy, $E_{\rm x}$, which is explicitly given by the last term on the right-hand side of Eq.~(\ref{eq:IN3}), and (ii) changing electron-electron repulsion and screening that gives rise to changing correlation energy, $E_\textbf{\rm c}$ as a result of other electrons' displacements. In view of $E_{\rm xc}$, we can observe that the Hartree-Fock theory properly takes $E_{\rm x}$ into account, but completely ignores $E_{\rm c}$. In contrast, DFT considers the exchange-correlation functional, $E_{\rm xc}[n]$, as it should be, compared to the Hartree-Fock theory. 

\texttt{NOTE 3}: In IET, the total energy is defined by the IET-Schr$\ddot{\rm o}$dinger equation~\cite{ADA1}, 
\begin {eqnarray}
&&i\hbar\frac{\partial \Psi(\textbf{r},t)}{\partial t} = \bigg[-\frac{\hbar^2}{2m}\nabla^2 + V_{\rm IET}\bigg]\Psi(\textbf{r},t) = H_{\rm IET}\Psi(\textbf{r},t) = (E_0 \pm \xi)\Psi(\textbf{r},t), \label{eq:IN14}
\end {eqnarray}  
where $E_0 \pm \xi = E$ is the real (true and unique) energy for a given quantum system (atom or molecule or solid or any quantum matter in between) for both degenerate (trivial because $\xi \rightarrow 0$) and non-degenerate energy levels~\cite{ADA1} and for quantum systems with energy-level crossings~\cite{ADA2}. Here, $E_0$ is the total energy at zero temperature and in the absence of any external disturbances, while $\xi$ is known as the ionization energy or the energy-level spacing~\cite{ADA5} where $+\xi$ is for electrons and $-\xi$ is for holes. The real eigenvalue, $E_0 \pm \xi$ cannot be obtained from the IET because both $V_{\rm IET}$ and $\Psi(\textbf{r},t)$ are unknown~\cite{ADA1}, and we did not bother to find them. To obtain the real value for $E_0 \pm \xi$, we need to know the real $\Psi(\textbf{r},t)$ and the real $V_{\rm IET}$. But, we can attempt to obtain the real energy eigenvalue, or close to the real one by making use of the relation, $E_0 \pm \xi = E = E^{\rm HF}$ (from Eq.~(\ref{eq:IN3})) or $E_0 \pm \xi = E = E[n]$ (from Eq.~(\ref{eq:IN13})) via the Hartree-Fock theory or DFT, respectively. 

The primary motivation for developing IET is not to calculate these energy eigenvalues, but to predict the changes to numbers associated to any physical quantities both qualitatively and quantitatively, without relying on wave functions because the real $\Psi(\textbf{r},t)$ remains, and will remain unknown, regardless whether one is able to know the real $E_{\rm xc}$, which is also unknown. Instead, our intention is to derive valid analytic equations from IET for all quantum matter, particularly, for the non-free-electron types by relying on the atomic energy-level spacings. These analytic equations can be exploited to develop theoretical models or to evaluate relevant physical quantities with respect to changing interaction strengths. This changing interaction strengths arise when one changes any (or all) of these numbers, electrons ($\zeta$), protons ($\tau$) and neutrons ($\eta$) in a given quantum system (recall \texttt{NOTE}: 1), which can be done either (i) by changing the number of atoms or the types of atoms in a given quantum low-energy system via chemical reactions, or (ii) via some nuclear reactions in a high-energy physical system. 

In order to achieve our objective to derive analytic functions, we make use of this approximation (also known as the ionization energy approximation) such that~\cite{ADA1}
\begin {eqnarray}
&& \bigg[-\frac{\hbar^2}{2m}\nabla^2 + V_{\rm IET}\bigg]\Psi(\textbf{r}) = H_{\rm IET}\Psi(\textbf{r}) \propto (E_0 \pm E_{\rm I})\Psi(\textbf{r}). \label{eq:IN15}
\end {eqnarray}  
The first part of this paper is to prove the proportionality (exists due to $\xi \propto E_{\rm I}$) given in Eq.~(\ref{eq:IN15}), where $\xi$ is the real energy-level spacing, while $E_\textbf{\rm I}$ denotes the unreal or the approximated energy-level spacing. We also have proven that $V_{\rm IET} \propto E_{\rm I}$ in accordance with atomic He, which can be understood from the following equations. First, we define  
\begin {eqnarray}
&& V_{\rm IET} = V_{\rm external}(r) + \tilde{V}_{\rm sc}(\textbf{r},\sigma) = -\frac{e^2}{4\pi\epsilon_0}\bigg[\frac{Z}{r_i} - \frac{1}{|\textbf{r}_i-\textbf{r}_j|}e^{-\sigma(r_i - r_j)}\bigg], \label{eq:IN16}
\end {eqnarray}  
where $Z$ is the atomic number, $\tilde{V}_{\rm sc}$ denotes the renormalized screened Coulomb potential. Recall here that $\textbf{R}_I = \textbf{R} = 0$ because we are evaluating an atom (He), while $\sigma$ is given by~\cite{ADA4}
\begin {eqnarray}
&&\sigma = \mu\exp{\bigg[-\frac{1}{2}\lambda\xi\bigg]}, \label{eq:IN17}
\end {eqnarray}  
where $\mu$ is the constant of proportionality, $\lambda = (12\pi\epsilon_0/e^2)a_{\rm B}$, $a_{\rm B}$ is the Bohr radius~\cite{ADA4}, and we can observe that $\tilde{V}_{\rm sc}(\textbf{r},\sigma) \propto E_{\rm I}$~\cite{ADA1}. The above renormalization procedure is exact and is based on the energy-level spacing renormalization group method~\cite{ADA5}. To see why the energy-level spacing renormalization method is exact, we give the following example. The pure exchange-energy functional (obtained by solving the last term on the right-hand side of Eq.~(\ref{eq:IN3}) using plane waves)~\cite{gell}, 
\begin {eqnarray}
&&E_{\rm x}[n] = -\frac{3}{4}\frac{e^2}{\pi}Nk_{\rm F} = -\frac{3e^2}{4}\bigg(\frac{3}{\pi}\bigg)^{1/3}\int \big[n(\textbf{r})\big]^{4/3}{\rm d}\textbf{r}, \label{eq:IN18}
\end {eqnarray}
while the renormalized $E_{\rm x}$ is given by
\begin {eqnarray}
&&\tilde{E}_{\rm x}[n] = -\frac{3}{4}\frac{e^2}{\pi}\tilde{N}\tilde{k}_{\rm F} = -\frac{3e^2}{4}\bigg(\frac{3}{\pi}\bigg)^{1/3}e^{-\frac{4}{3}\lambda\xi}\int \big[n(\textbf{r})\big]^{4/3}{\rm d}\textbf{r}, \label{eq:IN19}
\end {eqnarray}
where, $N = \int n(\textbf{r}){\rm d}\textbf{r}$, $\tilde{N} = e^{-\lambda\xi}\int n(\textbf{r}){\rm d}\textbf{r}$, $k_{\rm F} = (3\pi^2 n(\textbf{r}))^{1/3}$, $\tilde{k}_{\rm F} = e^{-(1/3)\lambda\xi}k_{\rm F}$ and $k_{F}$ denotes the Fermi wave number, which defines the Fermi surface in momentum-space. Reference~\cite{ADA4} contains the proofs and details on the above renormalization where $\tilde{n} = ne^{-\lambda\xi}$. The unrenormalized variable does not carry a tilde, and Eq.~(\ref{eq:IN19}) is exact by noting that when $\xi \rightarrow 0$, the electrons transform from being strongly-correlated to free electrons, while the strength of interaction increases with increasing $\xi$. In addition, $E_{\rm x}$ decreases with increasing $\xi$ as it should be because the outer-most electron will tend to stay as far away as possible from the core electrons if the energy-level spacing between the outer-most and the core electrons is large. Here, large $\xi$ also means a reduced screening effect, which will lead to a large electron-electron Coulomb repulsion~\cite{ADA1}.

The above renormalization procedure can also be related to Shankar renormalization technique reported in Refs.~\cite{shank1,shank2,shank3}, and the relevant proofs are given in Ref.~\cite{ADA5}. However, the Shankar renormalizer, $\Lambda_{\rm Shankar}$ is wave-number ($k$) dependent, while $\Lambda_{\rm IET}$ is $\lambda\xi$ dependent, and consequently, $\Lambda_{\rm IET}$ is dimensionless. This also means that (due to $\xi$) $\Lambda_{\rm IET}$ measures the interaction strengths with respect to the numbers and types of atoms in a given quantum matter because $\xi$ is unique for each different quantum system (due to different numbers of electrons ($\zeta$), protons ($\tau$) and neutrons ($\eta$)). This is exactly what we have wanted, and we have gotten it. What remains to be proven rigorously is $\xi \propto E_{\rm I}$ where this proportionality has been physically (in terms of probability) shown to be valid in Refs.~\cite{ADA7,ADA8}.

In summary, we have noted that the Hartree-Fock theory considers the electron-electron, electron-ion and the exchange interactions correctly, but completely ignores the correlation between electrons because each of these terms are isolated. For instance, changes in electron-electron repulsion is independent of the other two interactions (electron-ion and exchange), and vice versa. Therefore, the Hartree-Fock theory is suitable for non-interacting and weakly interacting systems. On the other hand, DFT does take the electron correlation effect into account by means of $E_{\rm xc}[n]$, but it is unable to handle $E_{\rm xc}[n]$ properly because $E_{\rm xc}[n]$ needs to be approximated, namely, LDA, LSDA and GGA by writing $E_{\rm xc}[n] = E_{\rm x}[n] + E_{\rm c}[n]$. In addition, $E_{\rm c}[n]$ is an unknown functional, and therefore, we do not know how and why $E_{\rm x}[n]$ can change when the strength of electronic correlation increases or decreases. For example, see Eq.~(\ref{eq:IN19}) to observe why and how the electron-electron correlation (captured by $\xi$) creates an additional constraint on $E_{\rm x}[n]$, such that, increasing repulsion between electrons (increasing $\xi$) gives rise to a decreasing $E_{\rm x}[n]$. Their dependence (between $E_{\rm x}[n]$ and $E_{\rm c}[n]$) can only be known if we know the analytic functional for $E_{\rm c}[n]$.

Obtaining the above-mentioned functionals are not easy. Consequently, we made use of the ionization energy theory to evaluate the changing interaction strengths by first acknowledging the existence of real (unique) energy-level spacings ($\xi$) for each quantum matter. This (due to changing interaction strength) also implies the possibility to obtain or define different classes of quantum matter. Subsequently, we can then use the energy-level spacing renormalization group method to derive the relevant renormalized theoretical models and analytic functions for a given solid or molecule. Almost all (if not all) of the important interactions (electron-electron, electron-ion, and spin-exchange) in atoms, also exist in solids and molecules. This statement can only become stronger if we compare between a solid or a molecule and their respective constituent atoms. This means that, logically one has $\xi_{\rm solids} \propto \xi_{\rm atoms}$ or $\xi_{\rm molecules} \propto \xi_{\rm atoms}$ due to the above similar types of interactions. Here, we have assumed that the many-body interaction due to different crystal or molecular structures is a constant. 

However, this many-body contribution can also be a non-trivial function (not a constant), which implies that these crystal- or molecular-structure effects may give rise to a QPT. In particular, non-constant crystal- or molecular-structure effects or due to some other external disturbances can cause $\xi$ to vary significantly. For example, when $\xi \rightarrow 0$ or $\xi \rightarrow$ constant or $\xi \rightarrow$ finite value (not a constant) or $\xi \rightarrow \infty$, and such changes will lead us to different types of trivial ($\xi \rightarrow 0$ and $\xi \rightarrow \infty$) and non-trivial quantum phase transitions (QPT) in solids and molecules. Hence, it all comes down to proving $\xi_{\rm solids} \propto \xi_{\rm atoms}$, $\xi_{\rm molecules} \propto \xi_{\rm atoms}$, and also proving the existence of energy-level spacing physical category for atoms that can be used to build different classes of quantum matter. The existence of physical categories will be proven in the second part of this paper, after proving the above proportionality (the ionization energy approximation). 

\section{The ionization energy theory uncovered}

Our objective here is to establish the correctness of the so-called ionization energy theory (IET) such that it does not violate established logical foundation and mathematics. The necessary step required to achieve this is to develop rigorous proofs uncovering the ionization energy theory such that its mathematico-logical foundation can be proven to be free of all logical and technical errors, as well as to expose any \textit{ad hoc} assumption, if there is any.

Here, we list all the propositions and definitions required to properly construct the ionization energy theory. We will prove all propositions, and will attempt to prove all statements. However, any unproven statements should be regarded as conjectures and/or physical expositions, and such statements are not needed for the logical construction of the ionization energy theory. Statements or equations that are used here, but have been explicitly proven elsewhere are referred to the appropriate references. The proofs here are developed such that they are not only rigorous but also contain, where appropriate, expositions on why and how these proofs are related to physical and chemical systems. We also wish to inform you that excessive use of symbols has been deliberately avoided for all possible cases, unless a compact form is required. It is to be noted here that particular care has been taken in constructing the sentences in which, the word \texttt{OR} in a sentence simply means either this or that, not both, whereas, the logical \texttt{OR} is symbolically denoted by $\vee$ to mean, either this or that or both. In sentences, we use \texttt{AND/OR} to denote $\vee$. Other logical notations and symbols, namely, \texttt{AND} ($\wedge$), \texttt{IMPLIES} ($\Rightarrow$), \texttt{NOT} ($\neg$) and \texttt{IF AND ONLY IF} ($\Leftrightarrow$) are valid both in mathematical forms and in sentences. 

The ionization energy theory is related to the chemical elements listed in the periodic table of chemical elements, and therefore, some basic properties of these elements and their connections to the energy levels are briefly explained below. However, how and why these energy levels can be associated to the ionization energies are given at a later stage, namely, in \ref{D13}, \ref{P14} and \ref{D15} because the technical reasons are not trivial.   

All the chemical elements in the periodic table are arranged in accordance to their atomic numbers $Z$, and each chemical element has unique energy levels, arising from the different numbers of protons ($\tau$), neutrons ($\eta$) and electrons ($\zeta$), and the physical interactions among them. Each electron either has spin-up ($\hbar/2$), or -down ($-\hbar/2$) where $\{\eta\} \in \mathbb{N}$, $\{\tau,\zeta\} \in \mathbb{N}^*$ and $\hbar = h/2\pi$, $h$ is the Planck constant. Here, atomic hydrogen (H), H$_2^+$ (molecular ion) and H$_2$ (molecule) all have zero neutron, and therefore, $\mathbb{N}$ and $\mathbb{N}^*$ are the sets of natural numbers, including and excluding zero, respectively. Moreover, $Z = \tau = \zeta$ for any neutral chemical element (or atom), molecule or compound (contains atoms of the order of Avogadro number, 10$^{23}$, though this is not necessarily true for nanoparticles). In particular, the number of atoms in a given nanoparticle can be of the order of 10 to 10$^3$ atoms, for example, see Refs.~\cite{ja,ja2,ja3,ja4}. When the chemical elements or molecules combine to form compounds \textit{via} chemical reactions, then these compounds too, have unique energy-level spacings ($\xi \geq 0 \wedge \{\xi\} \in \mathbb{R}^+$) such that

\begin{proposition}
\label{P1}
One can form any number of sets from these compounds to prove the existence of different classes of quantum matter.
\end{proposition}
\begin{proof}
See \ref{P21} and \ref{P22}.
\end{proof}

To prove \ref{P1}, one needs to invoke another proposition due to the ionization energy theory (IET) [see equations (5), (24) and (26) in Ref.~\cite{ADA1}],

\begin{proposition}
\label{P2}
If the formed compounds are of non free-electronic systems, then the energy-level spacing of such a system is 
proportional to their constituent atomic energy-level spacing.
\end{proposition}
\begin{proof}
See \ref{P6} and \ref{P14}.
\end{proof}

Here, $\mathbb{R}^+$ is the set of positive real numbers including zero, and physically, \ref{P2} means that $\xi$ (also known as the ionization energy) of a newly formed compound is proportional to the energy-level spacing of its constituent chemical elements (atomic $\xi$). Prior to proving \ref{P1} and \ref{P2}, we need additional propositions and a series of definitions, which then can be used to construct the ``energy-level spacing'' physical categories, required to develop the proofs for \ref{P1} and \ref{P2}. For clarity, each definition will end with a $\blacksquare$, while the Halmos tombstone, $\square$ is used to denote the end of a proof. We now start with the logical construction of the ionization energy theory.

\begin{definition}
\label{D3}
The energy-level spacings (the quantized and discrete energy-level differences) for atoms, molecules and compounds are denoted by $\xi_{\rm atom}$ {\rm (}e.g., $\xi_{\rm H}${\rm )}, $\xi_{\rm molecule}$ {\rm (}e.g., $\xi_{\rm H_2}${\rm )} and $\xi_{\rm compound}$ {\rm (}e.g., $\xi_{\rm YBa_2Cu_3O_7}${\rm )}, respectively $\blacksquare$
\end{definition}

\begin{proposition}
\label{P4}
$(a)$ $\xi_{\rm atom}$ is unique for any atom. $(b)$ $\xi_{\rm ion}$ is unique for any isolated ion with at least one bounded electron.
\end{proposition} 

\begin{proof} 
$(a)$ For an atomic hydrogen, one has $\xi_{\rm H} = f_{\rm H}(\zeta_{=\rm 1},\eta_{=\rm 0},\tau_{=\rm 1}) = f_{\rm H}$, which is unique because $f_{\rm H}(\zeta_{=\rm 1},\eta_{=\rm 0},\tau_{=\rm 1})$ is unique due to the unique numbers of electrons, neutrons and protons in each atom, and these unique numbers imply that the physical interactions among these particles are also unique. For a helium atom $\xi_{\rm He} = f_{\rm He}(\zeta_{=\rm 2},\eta_{=\rm 2},\tau_{=\rm 2}) = f_{\rm He}$. Again, the uniqueness of $f_{\rm He}$ is guaranteed by the unique numbers of $\zeta, \eta$ and $\tau$. $(b)$ It is straightforward to check that $f_{\rm He^+}$ is unique, and this is sufficient.
\end{proof}

We will provide the details on $f(\zeta,\eta,\tau)$ in \ref{P10}. This is to make sure the construction of IET is systematic and logical in its presentation. Anyway, what is unambiguously certain here and now is that each $f(\zeta,\eta,\tau)$ is uniquely defined satisfying the unique numbers of protons and neutrons confined in the nucleus (\textit{via} the so-called strong force), and the number of bounded electrons (due to the Coulomb force) revolving around the nucleus. Therefore, the uniqueness of $\xi$ is due to the unique function, $f(\zeta,\eta,\tau)$. Alternatively, the uniqueness of $\xi_{\rm atom}$ and $\xi_{\rm ion}$ are physically self-evident from the periodic table of chemical elements. However, both the position and the momentum of any bounded electron, in any quantum matter, can never be determined simultaneously. This is famously known as the Heisenberg uncertainty principle [see equation (3.63) in Ref.~\cite{DJG1}]. 

Here, we introduce some useful scientific terminologies--- neutral chemical elements (from the periodic table of chemical elements) are also known as atoms, isolated compounds of several atoms are called molecules, while the charged atoms are denoted by cations (positively charged ions), or anions (negatively charged ions), or simply ions in general. Molecular ions are simply charged molecules (positively or negatively charged). Plasmas on the other hand, consist of charged gas particles. In other words, plasmas can consist of isolated charged-molecules and/or -atoms and/or -nanoparticles (also known as dust particles) and/or any combination of them. Here, nanoparticles have atoms of the order of much less than 10$^{23}$, which cannot be considered as molecules due to some subtle reasons that are irrelevant here. By definition, a proton is a positively ($+e$) charged sub-atomic particle, while an electron is a negatively ($-e$) charged sub-atomic particle. A neutron on the other hand, is a neutral sub-atomic particle. 

However, the different charge property alone (namely, positive, negative and neutral) among these sub-atomic particles does not give a complete physical description for these sub-atomic particles. For example, an electron is a fundamental particle without any constituent sub-particle, whereas, a proton and a neutron have their own distinct constituent sub-particles. More detailed scientific terminologies are not required to understand what is written here. 

\begin{proposition}
\label{P5}
The energy-level spacings defined in $\ref{D3}$ are unique for each different system, be it for atoms, molecules, or solids.
\end{proposition}

\begin{proof} 
We first assume \ref{P6}$(a)$ is true, and defer the proof for \ref{P6}$(a)$ after \ref{D15}. The uniqueness of each atom and each isolated ion with at least one bounded electron have been proven earlier (see \ref{P4}).

\begin{proposition}
\label{P6}
$(a)$ For any general quantum matter or system $($excluding atoms and isolated ions$)$, 
\begin {eqnarray}
\xi_{\rm system} \propto E^{\rm system}_{\rm I} &=& \bigg[\sum_ql_q\sum_{\nu}\frac{1}{\nu}f_{q,\nu}(\zeta,\eta,\tau)\bigg]_{\rm cation} + \nonumber \\&& \bigg[\sum_ql_qf_{q,\nu = 1}(\zeta,\eta,\tau)\bigg]_{\rm anion}, \label{eq:1}
\end {eqnarray}  
where $q$ and $l_q$ represent the types of chemical elements contained in a non-atomic system, and the composition of each chemical element in a particular non-atomic system, respectively. Here, $\nu$ is the number of outer electrons {\rm (}also known as valence electrons{\rm )} that are responsible for any chemical bonding and for any physical interaction in any system such that $\nu \in \mathbb{N}^*$ and $\nu \leq \zeta$.

$(b)$ For atoms and isolated ions, $Eq.~(\ref{eq:1})$ is exact such that $\xi_{\rm atom} = E^{\rm atom}_{\rm I} = f_{\rm atom}(\zeta, \eta, \tau)$ and $\xi_{\rm ion} = E^{\rm ion}_{\rm I} = f_{\rm ion}(\zeta_{\rm ion}, \eta, \tau)$, respectively, where $\tau < \zeta_{\rm ion}$ for any isolated anions, while $\tau > \zeta_{\rm ion}$ for any cations.
\end{proposition}

The above-stated proportionality ($\propto$) actually defines the ionization energy approximation. The complete proof for the ionization energy approximation is given after \ref{D15}. For example, the three proposals in \ref{P6}, \ref{P14} and \ref{P16} give the complete proofs for the ionization energy approximation. One should also note here that the cations are electron-donors, while the anions are electron-acceptors.

\begin{definition}
\label{D7}
All real energy-level spacings are always denoted by $\xi$. The unreal ones are denoted by $E^{\rm system}_{\rm I}$ such that the term ``unreal'' means $\xi_{\rm system} \neq E^{\rm system}_{\rm I}$, but implies $\xi_{\rm system} \propto E^{\rm system}_{\rm I}$ $($from \ref{P6}$)$ $\blacksquare$
\end{definition}

It is sufficient for us to prove \ref{P5} by considering a solid and a molecule. This proof can be easily verified for other systems. The proposition \ref{P4} given earlier is sufficient for any atomic system and isolated ions. For a particular solid, say Si$^{4+}$O$^{2-}_2$, one can write (from \ref{P6} and \ref{D7}) 
\begin {eqnarray}
&& \xi_{\rm SiO_2} \propto E^{\rm SiO_2}_{\rm I} = x(1/4)\big[f_{\rm Si^{+}}(\zeta_{=\rm 13},\eta_{=\rm 14},\tau_{=\rm 14}) + f_{\rm Si^{2+}}(\zeta_{=\rm 12},\eta_{=\rm 14},\tau_{=\rm 14}) + \nonumber \\&& f_{\rm Si^{3+}}(\zeta_{=\rm 11},\eta_{=\rm 14},\tau_{=\rm 14}) + f_{\rm Si^{4+}}(\zeta_{=\rm 10},\eta_{=\rm 14},\tau_{=\rm 14})\big] + yf_{\rm O^+}(\zeta_{=\rm 7},\eta_{=\rm 8},\tau_{=\rm 8}). \label{eq:2}
\end {eqnarray}
Since the numbers of electrons, neutrons and protons are unique for each isolated cation and anion, one can simplify the notations to write, $E_{\rm I}^{\rm SiO_2} = x(1/\nu)\sum_{\nu}f_{\rm Si^{\nu +}} + yf_{\rm O^+} = x(1/\nu)\sum_{\nu}\xi_{\rm Si^{\nu +}} + y\xi_{\rm O^+}$. Here, $x = 1$ and $y = 2$ denote the ratio of Si and O atoms in SiO$_2$ and $\{x,y\} \in \mathbb{R}_*^+$ where the set of positive numbers excluding zero is denoted by $\mathbb{R}_*^+$. Knowing that $x(1/\nu)\sum_{\nu}f_{\rm Si^{\nu +}}$ and $yf_{\rm O^+}$ are individually unique (from \ref{P4}), we can conclude that $E_{\rm I}^{\rm SiO_2}$ is also unique. For a H$^+_2$O$^{2-}$ molecule, we have $\xi_{\rm H_2O} \propto E_{\rm I}^{\rm H_2O} = if_{\rm H^+}(\zeta_{=\rm 0},\eta_{=\rm 0},\tau_{=\rm 1}) + jf_{\rm O^+}(\zeta_{=\rm 7},\eta_{=\rm 8},\tau_{=\rm 8}) = if_{\rm H^+} + jf_{\rm O^+} = 2\xi_{\rm H^+} + 1\xi_{\rm O^+}$ where $\{i,j\} \in \mathbb{N}^*$. Therefore, $E_{\rm I}^{\rm H_2O}$ has to be unique from the uniqueness of $if_{\rm H^+}$ and $jf_{\rm O^+}$.
\end{proof}
 
\begin{definition}
\label{D8}
From \ref{P4}, each atom forms a singleton due to the uniqueness of $\xi_{\rm atom}$ $\blacksquare$
\end{definition}
 
For example, $\{\rm H\}$, $\{\rm He\}$, $\{\rm Be\}$, $\{\chi_i\}$, $\cdots$, $\{\chi_j\}$ such that $\{\chi\}$ is called singleton $\chi$, and represents any stable or unstable (with relatively very short lifetimes) chemical element where $\{i,j\} \in \mathbb{N}^*$, $i < j$ and $j$ is finite.

\begin{definition}
\label{D9}
Physically, any isolated ion without any bounded electron cannot be part of any system consisting of more than one chemical element. In other words, no system with more than one chemical element can exist without any bounded electron due to repulsive Coulomb interaction between two positively charged nucleus $\blacksquare$
\end{definition}

However, it is to be noted here that any interaction between two nuclei may give rise to a new chemical element, and it involves nuclear reactions [see \ref{P21}, after Eq.~(\ref{eq:22})], which will be covered when we discuss relations between singletons. Moreover, the cations and anions considered in \ref{P6} is for systems with electrons and these systems are strictly neutral ($\zeta = \tau$). The systems (SiO$_2$ and H$_2$O) stated in \ref{P5} are neutral because the electrons contributed by the respective cations are still bounded within that particular system such that there is a non-zero interaction between the cation and that particular electron.

\begin{proposition}
\label{P10}
Following $\ref{D9}$, $\xi$ is zero for any isolated cation in the absence of electrons such that there is a zero interaction between the isolated cation and any electron, and therefore, $f_{\rm cation}(\zeta_{=\rm 0},\eta_{\geq\rm 0},\tau_{>\rm 0}) = 0$.
\end{proposition}
 
\begin{proof} For an isolated and electronless cation (H$^+$), the function $f_{\rm H^{+}}(\zeta_{= 0},\eta_{=\rm 0},\tau_{=\rm 1}) = 0$ due to zero bounded electron, and therefore, $\xi_{\rm H^{+}} = 0$. For He$^{2+}$, one has $f_{\rm He^{2+}}(\zeta_{=\rm 0},\eta_{=\rm 2},\tau_{=\rm 2}) = 0$ and consequently $\xi_{\rm He^{2+}} = 0$, and one can go on and prove $f_{\rm cation}(\zeta_{=\rm 0},\eta_{\geq\rm 0},\tau_{>\rm 0}) = 0$ for all the chemical elements (stable or unstable) in the periodic table, if and only if the interaction between a cation and any electron is always zero.
\end{proof}

\begin{proof} This is an alternative proof with embedded quantum physical notions, and therefore, can be regarded as an exposition in quantum physics. To see why $f_{\rm cation}(\zeta_{=\rm 0},\eta_{\geq\rm 0},\tau_{>\rm 0}) = 0$, we just need to prove that the total electronic energy in a particular electronless cation is zero. It is sufficient to consider a single atomic H and its only cation H$^+$. The total electronic energy for the one bounded electron in an atomic H is the sum of the electronic kinetic energy, and the electron-nucleus attractive Coulomb potential energy. The total electronic energy can be calculated from the non-relativistic and time-independent quantum mechanical formalism (Schr$\ddot{\rm o}$dinger representation),        
\begin{eqnarray}
&&\bigg[-\frac{\hbar^2}{2m}\nabla^2 + V \bigg]_{\rm H}\varphi^{\rm H}_1 = \big[E_0 + f(\zeta_{=\rm 1},\eta_{= \rm 0},\tau_{= \rm 1})\big]_{\rm H}\varphi^{\rm H}_1 = \big[E_0 + \xi\big]_{\rm H}\varphi^{\rm H}_1. \label{eq:3}
\end{eqnarray}
Here, $(\hbar^2/2m)\nabla^2$, $V \propto [(+e)(-e)/r_1] + [(+e)(-e)/r_2]$ and $E_0 + \xi$ are the kinetic energy operator, Coulombic potential energy operator and the total energy eigenvalue, respectively, for the bounded electron. The charges, $+e$ and $-e$ refer to the charge of a proton and an electron, respectively, while $r_{1,2}$ is the usual electron coordinate, relative to protons 1 and 2, respectively. Moreover, $\varphi^{\rm H}_1$ represents the wave function for the bounded electron, $E_0$ is defined to be the total energy in the absence of all types of perturbations, including temperature, $T$ (that is $T$ = 0 Kelvin). The mass ($m$) of the bounded electron is much smaller compared to the mass of the nucleus ($M$) such that $M \gg m$, and thus, we can assume the nucleus is static relative to the bounded electron. This assumption is known as the Born-Oppenheimer approximation [see equation (49) in Ref.~\cite{ADA2}], and it is invoked here entirely for technical convenience because we are concerned only with the existence of $\xi$ due to boundedness. Thus, any correction to $\xi$ due to vibrating nucleus or other subtle interactions can always be assumed to be taken into account (see the logical proof below \ref{P10}). As a consequence, Eq.~(\ref{eq:3}) is independent of the kinetic energy of the nucleus. The total energy, $E_0 + \xi$ is the energy needed to remove the bounded electron to infinity or to any finite distance (if any) such that there is no interaction between H$^{+}$ and the removed electron. In the absence of any bounded electron, Eq.~(\ref{eq:3}) reads  
\begin{eqnarray}
\bigg[-\frac{\hbar^2}{2m}\nabla^2 + V \bigg]_{\rm H^+}\varphi^{\rm H^+}_0 &=& \big[E_0 + f(\zeta_{=\rm 0},\eta_{= \rm 0},\tau_{= \rm 1})\big]_{\rm H^+}\varphi^{\rm H^+}_0 \nonumber \\&=& \big[E_0 + \xi \big]_{\rm H^+}\varphi^{\rm H^+}_0 = 0, \label{eq:4}
\end{eqnarray}
where $E_0 = 0 = f_{\rm H^+}(\zeta_{=\rm 0},\eta_{=\rm 0},\tau_{=\rm 1})$ because there is no bounded electron ($\zeta = 0$), and this is guaranteed by the fact that $[(\hbar^2/2m)\nabla^2]\varphi_0^{\rm H^+} = 0$ and $V\varphi_0^{\rm H^+} = 0$ because in the absence of any bounded electron, $E_0^{\rm H^+} = 0 = \xi^{\rm H^+} \wedge \varphi_0^{\rm H^+} = 0$ (implies $\varphi_0^{\rm H^+}$ and $\xi^{\rm H^+}$ cannot exist if there is not a single bounded electron).
\end{proof}

\begin{definition}
\label{D9(b)}
Proposition \ref{P10} implies that there exists a correspondence rule such that $\xi = f(\zeta,\eta,\tau) \rightarrow \varphi$ or $\varphi \rightarrow f(\zeta,\eta,\tau) = \xi$ $\blacksquare$
\end{definition}

In physics, $\varphi^{\rm H}_1$ is known as the electronic wave function, and the ``real'' wave function is never known for any real atoms or real systems or real ions with $\zeta \geq 1, \eta > 0, \tau \geq 1$, except for atomic hydrogen ($\zeta = \tau = 1, \eta = 0$). The wave functions for any real system (including for any real atom or ion, other than hydrogen) are generated by some educated guesses [see Statement 2 and Remark 5 in Ref.~\cite{ADA2}] such that the wave functions give the minimum total energy. The energy-minimization procedure follows the Rayleigh-Ritz variational principle [see equations (7.32) and (7.33) in Ref.~\cite{DJG1}]. By now, some of you may have guessed correctly why IET belongs to the energy-level spacing categories, which is due to the fact that IET never requires any knowledge on wave functions (\ref{D9(b)}). Instead, it relies on the uniqueness of the electronic energy-level spacings in the presence of a proton (atomic hydrogen), or a nucleus (atoms or ions) or nuclei (molecules or compounds). Here, \ref{P10} is also true within the relativistic and time-dependent quantum mechanics because these additional effects do not disturb the uniqueness of $\xi$ in any way. Moreover, we have excluded the concept of positively charged electron-like particles (known as holes) for simplicity. To incorporate the concept of holes into the proof, we just need to rewrite the total energy eigenvalue as $E_0 - \xi$. By definition, $E_0$ is a negative number (due to bounded electrons) ranges between 0 and $-\infty$. The minus sign in $- \xi$ implies $(E_0 - \xi) \rightarrow -\infty$ for holes (for free or unbounded holes, $(E_0 - \xi) = -\infty$), while for electrons, $(E_0 + \xi) \rightarrow 0$ (for free or unbounded electrons, $(E_0 + \xi) = 0$). The holes as positive charge carriers exist in doped semiconductors and in some free-electron metals. Precise definitions and proofs on the existence of holes in semiconductors are given elsewhere [see the conditions \texttt{A1} and \texttt{B2}, and the section, mathematical analysis in Ref.~\cite{ADA3}].

\begin{proposition}
\label{P11}
$(a)$ Each anion forms a singleton due to the uniqueness of $f_{\rm anion}(\zeta,\eta,\tau)$ such that $f_{\rm anion}(\zeta,\eta,\tau) = \xi_{\rm anion}$. $(b)$ Physically, a single-electron anion cannot exist.
\end{proposition}

\begin{proof} $(a)$ Following \ref{P4} and \ref{D8}, we can construct these singletons for any anion such that $\{\chi_{i}^{\alpha-}\}, \cdots, \{\chi_{i}^{\alpha'-}\}, \cdots, \{\chi_{j}^{\alpha-}\} \cdots, \{\chi_{j}^{\alpha'-}\}$ where $\alpha < \alpha'$ and $\{i,j,\alpha,\alpha'\} \in \mathbb{N}^*$. If $\alpha = \alpha' = 0$, then $\{\chi_i\}, \cdots, \{\chi_j\}$ denote the singletons that have the usual neutral atoms or chemical elements as elements, which have been defined earlier in \ref{D8}. $(b)$ Recall that the boundedness condition requires at least a proton (positively charged by definition) as a nucleus and a bounded electron (negatively charged by definition). Therefore, indeed a single electron anion cannot exist.
\end{proof}

\begin{proposition}
\label{P12}
Each cation forms a singleton due to the uniqueness of $f_{\rm cation}(\zeta,\eta,\tau)$ such that $f_{\rm cation}(\zeta,\eta,\tau) = \xi_{\rm cation}$ and $\zeta \geq 0$.
\end{proposition}

\begin{proof} Following \ref{P4} and \ref{D8}, we can construct the singletons for any cation such that $\{\chi_{i}^{\omega+}\}, \cdots, \{\chi_{i}^{\omega'+}\}, \cdots, \{\chi_{j}^{\omega+}\}, \cdots, \{\chi_{j}^{\omega'+}\}$ if and only if $\zeta > \omega$ and $\zeta > \omega'$ where $\omega < \omega'$ and $\{i,j,\omega,\omega'\} \in \mathbb{N}^*$. An isolated electronless cation gives $\xi_{\rm cation} = 0$ as proven in \ref{P10}.
\end{proof}

\begin{definition}
\label{D13}
$(a)$ The energy-level spacing, $\xi^{\rm first} = E_{\rm U} - E^{\rm first}_{\rm I}$. Here, $E_{\rm U}$ is the most lowest unoccupied energy-level, while the top most electron occupied energy-level, $E^{\rm first}_{\rm I}$ is one level below $E_{\rm U}$ such that $|E_{\rm U}| < |E^{\rm first}_{\rm I}| < |E^{\rm second}_{\rm I}| < \cdots < |E^{(\zeta)}_{\rm I}|$. These electronic energy levels are defined between $0$ and $-\infty$ due to their boundedness. $(b)$ The energy-level, $E^{\rm first}_{\rm I}$ is known as the first ionization energy $($with the smallest magnitude$)$, and so forth for other core electrons $\blacksquare$
\end{definition}

\begin{proposition}
\label{P14}
For any isolated atom or ion such that $\zeta > 0$, the energy-level spacing is given by $\xi^{\rm first} = (E_{\rm U} - E^{\rm first}_{\rm I}) = -E^{\rm first}_{\rm I}$. In compact form, $\xi^{\rm first} = -E^{\rm first}_{\rm I}$, $\xi^{\rm second} = -E^{\rm second}_{\rm I}, \cdots, \xi^{(\zeta)} = -E^{(\zeta)}_{\rm I}$.
\end{proposition}

\begin{proof}
It is immediately obvious that $E^{\rm first}_{\rm I}$, $E^{\rm second}_{\rm I}$, $\cdots$, $E^{(\zeta)}_{\rm I}$, and $E_{\rm U}$ are individually unique for any isolated atom or ion due to the uniqueness of $\xi$ from \ref{D3}, \ref{P4}, \ref{D13} and \ref{D15}.

\begin{definition}
\label{D15}
$(a)$ The first ionization energy $(E^{\rm first}_{\rm I})$ is the energy required to remove the first outer-most electron $($that has the smallest energy$)$ from any isolated chemical element or ion to a distance $r \rightarrow \infty$ such that there is no interaction between the ion and the removed electron. $(b)$ Any electron at $r \rightarrow \infty$ is not bounded to any nucleus, and therefore $E_{r \rightarrow \infty} = 0$ $\blacksquare$
\end{definition}

First, we normalize $E_{\rm U}$ such that all the atoms and ions with $\zeta > 0$ have the same exact reference point. This is possible if we take $E_{\rm U} \rightarrow E_{r \rightarrow \infty}$ such that $E_{\rm U} = E_{r \rightarrow \infty} = 0$. Due to this normalization, we can write (from \ref{D15}) $\xi^{\rm first} = (E_{\rm U} - E^{\rm first}_{\rm I}) = (E_{r \rightarrow \infty} - E^{\rm first}_{\rm I}) = (0 - E^{\rm first}_{\rm I}) = -E^{\rm first}_{\rm I}$, therefore, $\xi^{\rm first} = -E^{\rm first}_{\rm I}$ and so on for the second, third, and other core electrons.
\end{proof}

Of course, we could have normalized $E_{\rm U}$ such that $E_{\rm U} = E_{r \rightarrow |\textbf{r}|}$ where $|\textbf{r}|$ denotes a fixed arbitrary finite distance from the nucleus, but this will only lead us to the same problem of not being able to determine the magnitude of $E_{\rm U}$ because $E_{r \rightarrow |\textbf{r}|}$ is also unknown. Consequently, we can observe why $\xi^{(\zeta)} = -E_{\rm I}^{(\zeta)}$ for any atomic or ionic system for $\zeta > 0$. Hence, the proof for \ref{P14} is also a proof for \ref{P6}$(b)$.

\begin{proof} The following is a proof for the \textbf{Proposition \ref{P6}$(a)$} given earlier. For any isolated atom or ion such that $\zeta > 1$, there exist an exact averaged energy-level spacing for the case where there are more than one excited or polarized electrons to some higher energy-levels simultaneously. Hence, the averaged energy-level spacing for an arbitrary electron-donating chemical element, $\chi$ and its cation (after donating), $\chi^{\nu +}$ that has $\nu$ ($\nu < \zeta$) valence electrons polarized or excited to higher energies is given by 
\begin{eqnarray}
&\xi(\chi^{\nu +})& = \frac{1}{\nu}\big\{[E_{\rm U}(\chi) - E^{\rm first}_{\rm I}(\chi)] + \cdots + [E_{\rm U}(\chi) - E^{(\nu)}_{\rm I}(\chi)]\big\} \nonumber \\&& = \frac{1}{\nu}\big\{[E_{r \rightarrow |\textbf{r}|}(\chi) - E^{\rm first}_{\rm I}(\chi)] + \cdots + [E_{r \rightarrow |\textbf{r}|}(\chi) - E^{(\nu)}_{\rm I}(\chi)]\big\} \nonumber \\&& = -\frac{1}{\nu}\big\{E^{\rm first}_{\rm I}(\chi) + \cdots + E^{(\nu)}_{\rm I}(\chi)\big\}. \label{eq:5}
\end{eqnarray}
Using \ref{P4}, \ref{P14} and noting that $E_{\rm I} < 0$ due to boundedness (bounded electrons), one obtains
\begin{eqnarray}
&\xi(\chi^{\nu +})& = \frac{1}{\nu}\big\{\xi^{\rm first}(\chi) + \cdots + \xi^{\nu +}(\chi)\big\} \nonumber \\&& = \frac{1}{\nu}\big\{f(\zeta-1,\eta,\tau)(\chi^+) + \cdots + f(\zeta-\nu,\eta,\tau)(\chi^{\nu +})\big\} \nonumber \\&& = \bigg[\sum_{\nu}\frac{1}{\nu}f_{\nu}(\zeta,\eta,\tau)\bigg]^{\rm cation}_{\rm \chi^{\nu +}}. \label{eq:6}
\end{eqnarray}

\begin{proposition}
\label{P16}
Certain neutral atoms attract electrons from its neighboring cations to form an anion $(\chi^{\nu -})$ and therefore, its first ionization energy $(\xi(\chi^{+}) = E^{\rm first}_{\rm I})$ alone determines its ability to attract any not-fully bounded electron, or electrons from any neutral electron-donating neighboring atoms such that large $E^{\rm first}_{\rm I}$ implies a stronger attractive interaction between the electron-accepting nucleus and the not-fully bounded electron, or an electron originating from another electron-donating neutral atom $($after donating, this atom will form a cation$)$.
\end{proposition}

\begin{proof}
It is sufficient to consider the attractive interaction between an anion-forming neutral atom and a nearby not-fully bounded electron. From \ref{D15}, large $E^{\rm first}_{\rm I}$ also implies a large energy is required to remove, or to excite, or to polarize the outer-most electron due to strong electron-nucleus attractive interaction. Thus, we need to prove that there is also a stronger attractive interaction between the nucleus and the not-fully bounded electron. If this is true, then there must be a strong electron-electron repulsion between the outer-most electron and the not-fully bounded electron due to this stronger attractive interaction. For example, the stronger electron-nucleus attractive interaction attracts both electrons (the outer-most and the not-fully bounded one) equally toward the nucleus. However, this stronger attraction also implies a stronger electron-electron repulsion between these two electrons, and therefore, there exists a large energy-level spacing between these two electrons. Therefore, the total potential energy, $V_{\rm total}$ becomes smaller, and all we need to do here is to show that this total potential energy is indeed smaller for large $E^{\rm first}_{\rm I}$. Here, $V_{\rm total}$ consists of two main terms, one of the terms refers to the bare Coulomb potential energy between the nucleus ($Z(+e)$) and the two electrons that carry the same charge, $-e$ each. Recall that these two electrons refer to the outer-most and the not-fully bounded electrons. The second term originates from the electron-electron repulsion, which needs to be renormalized to account for the screening effect. This renormalization also accounts for the different screening strengths due to different atoms. See \ref{D17} to understand what constitutes the screening effect. Thus, the renormalized screened Coulomb potential energy is between the outer-most electron (from the anion-forming atom) and the not-fully bounded electron. We can now write the total potential energy, 
\begin{eqnarray}
&V_{\rm total}& = \sum_{i\neq j}\frac{Z(+e)(-e)}{4\pi\epsilon_0r_i} + \frac{1}{2}\bigg[\frac{(-e)(-e)}{4\pi\epsilon_0|\textbf{r}_i-\textbf{r}_j|}\bigg]\exp[{-\sigma (r_i + r_j)}] \nonumber \\&& = -\frac{e^2}{4\pi\epsilon_0} \bigg[\sum_{i\neq j}\frac{Z}{r_i} - \frac{1}{2}\frac{1}{|\textbf{r}_i-\textbf{r}_j|}\exp[{-\sigma (r_i + r_j)}]\bigg], \label{eq:7}
\end{eqnarray}
\begin{eqnarray}
\sigma = \mu\exp\bigg[{-\frac{1}{2}\lambda\xi}\bigg]. \label{eq:8}
\end{eqnarray}
The atomic number, $Z = \tau$, $e$ denotes the magnitude of an electron charge, recall that $+e$ and $-e$ refer to the charge of a proton and an electron, respectively, while $\epsilon_0$ is the permittivity of free space. The constant, $\lambda = (12\pi\epsilon_0/e^2)a_{\rm B}$ where $a_{\rm B}$ is the Bohr radius of an atomic hydrogen. For convenience, one can assume $\textbf{r}_{i = 1}$ is the coordinate of the outer-most electron, while $\textbf{r}_{j = 2}$ is the coordinate of the not-fully bounded electron appears in the vicinity of the anion-forming neutral atom. The factor 1/2 in front of the second (renormalized) term in Eq.~(\ref{eq:7}) is to avoid $ij \neq ji$, and $\mu$ is the screening constant of proportionality. The renormalized second term in Eq.~(\ref{eq:7}) has been derived earlier [see equation (11) in Ref.~\cite{ADA4}] based on the renormalization procedure [see equations (2.11), (2.23) and (2.25) in Ref.~\cite{ADA5}].

\begin{definition} 
\label{D17} 
Screening is a well-known physical concept in atomic physics where an outer-most electron of a many-electron atom interact with the positively charged nucleus $(+Ze)$ with an effective nucleus charge less than $+Ze$ as seen by this outer-most electron. In other words, the outer-most electron is screened by the core electrons. In the absence of this screening
\begin{eqnarray}
\lim_{\xi\rightarrow\infty}V_{\rm total} = -\frac{e^2}{4\pi\epsilon_0} \bigg[\sum_{i \neq j}\frac{Z}{r_i} - \frac{1}{2}\frac{1}{|\textbf{r}_i-\textbf{r}_j|}\bigg]. \label{eq:9}
\end{eqnarray}    
This is also known as the bare Coulomb potential energy. In contrast, for a completely screened outer-most electron   
\begin{eqnarray}
\lim_{\xi\rightarrow 0}V_{\rm total} = -\frac{e^2}{4\pi\epsilon_0} \bigg[\sum_{i\neq j}\frac{Z}{r_i} - \frac{1}{2}\frac{1}{|\textbf{r}_i-\textbf{r}_j|}\exp[{-\mu(r_i + r_j)}]\bigg]. \label{eq:10}
\end{eqnarray}
Here, $\lim_{\xi \rightarrow \infty}V_{\rm total}$ gives the unscreened Coulomb potential energy with maximum electron-electron repulsion. On the other hand, $\lim_{\xi \rightarrow 0}V_{\rm total}$ denotes the so-called Thomas-Fermi free-electron screening [see equation (1.13) in Ref.~\cite{ADA5}]. The free electrons are still bounded within a system, but they do not individually nor collectively bounded to any atom. In other words, they are collectively confined within a system due to the many-body (many-electron and many-nucleus) potential [see equation (2.4) in Ref.~\cite{AM1}] $\blacksquare$ 
\end{definition}

Having defined that, it is now straightforward to see that $|\lim_{\xi \rightarrow \infty}V_{\rm total}| < |\lim_{\xi \rightarrow 0}V_{\rm total}|$ as required to complete the proof for \ref{P16} where $\xi = -E^{\rm first}_{\rm I}$ (from \ref{P14}), $E^{\rm first}_{\rm I} < 0$ (due to boundedness), and therefore $\xi > 0$ where $E^{\rm first}_{\rm I} = 0 = \xi$ imply unboundedness. One should also note here that $E^{\rm first}_{\rm I}$ is nothing but the energy of the outer-most bounded electron.
\end{proof}

\begin{definition} 
\label{D18} 
$(a)$ A free-electron solid gives rise to a metallic property of which, its resistivity (resistance to electron flow) is almost entirely due to electron-electron and electron-phonon scattering. Phonons are vibrating ions that scatter the electrons, and this is the only direct role played by the ions. The indirect role is to observe charge neutrality. In addition, these electrons are free to the extent that they are treated as non-interacting electron gas [see equations (1.1), (1.6), (2.4) and (2.25), and Chapter 3 in Ref.~\cite{AM1}]. Within IET, a precise definition can be constructed, namely,

$(b)$ Free-electron metals simply require $\xi = 0$, while the weakly interacting electrons (with some electron-electron and electron-phonon interaction) leads to Fermi liquid, which embeds the condition, $\xi \neq 0$ but its magnitude is an irrelevant constant [see equations (2.19) and (2.20) in Ref.~\cite{ADA6}] $\blacksquare$
\end{definition}

\begin{proposition} 
\label{P19} 
The energy-level crossings in free-electron metals are infinite. This means that the energy gap is definitely zero throughout the {\rm \textbf{k}}-space and therefore, $\xi = 0$. 
\end{proposition}

\begin{proof}
Let us assume the energy-level crossings are finite, this means that the energy gap cannot be zero throughout the \textbf{k}-space. If we further assume $\xi$ is zero (from \ref{D18}$(b)$), then these two assumptions will contradict with each other. In particular, zero energy gap for all \textbf{k} is true if and only if $\xi$ is zero [see equation (2.1) in Ref.~\cite{ADA6}] and [see equations (39) and (49) in Ref.~\cite{ADA2}]. Therefore, the statement--- finite energy-level crossings can give rise to $\xi = 0$ is false, and consequently, finite energy-level crossings cannot lead to any free-electron metal.  
\end{proof}

We now continue with our proof for the proposition given in \ref{P6}$(a)$. Similar to Eqs.~(\ref{eq:5}) and~(\ref{eq:6}), the averaged energy-level spacing for an arbitrary electron-accepting chemical element, $\chi$ and its anion (after accepting), $\chi^{\nu -}$ is solely determined by $\xi(\chi^{+})$ (from \ref{P16}) where 
\begin{eqnarray}
&&\xi(\chi^{+}) = [E_{\rm U}(\chi) - E^{\rm first}_{\rm I}(\chi)] = [E_{r \rightarrow |\textbf{r}|}(\chi) - E^{\rm first}_{\rm I}(\chi)] = -E^{\rm first}_{\rm I}(\chi). \label{eq:11}
\end{eqnarray}
Using \ref{P4}, \ref{P14}, and noting that $E_{\rm I} < 0$ due to boundedness (bounded electrons), one obtains
\begin{eqnarray}
&&\xi(\chi^{+}) = \xi^{\rm first}(\chi) = f(\zeta-1,\eta,\tau)(\chi^+) = \bigg[f_{\nu = 1}(\zeta,\eta,\tau)\bigg]^{\rm anion}_{\rm \chi^{+}}. \label{eq:12}
\end{eqnarray}
Assume a hypothetical system, $\texttt{A}^{\nu_{\texttt{A}} +}_{a}\texttt{B}^{\nu_{\texttt{B}} -}_{b}\texttt{X}^{\nu_{\texttt{X}} +}_{x}\texttt{Y}^{\nu_{\texttt{Y}} -}_{y}$ where \texttt{B} and \texttt{Y} are anions, attracting $\nu_{\texttt{B}}$ and $\nu_{\texttt{Y}}$ electrons, respectively, while the cations donate $\nu_{\texttt{A}}$ and $\nu_{\texttt{X}}$ electrons, respectively, such that $\nu_{\texttt{A}} + \nu_{\texttt{X}} - \nu_{\texttt{B}} - \nu_{\texttt{Y}} = 0$ (for charge neutrality). Furthermore, $a$, $b$, $x$ and $y$ are the respective concentrations for the hypothetical chemical elements, \texttt{A}, \texttt{B}, \texttt{X} and \texttt{Y} in the above-stated hypothetical system. The average energy-level spacing for this system is given by (from Eqs.~(\ref{eq:6}) and~(\ref{eq:12})) 
\begin{eqnarray}
E_{\rm I}^{\rm system} &=& a\bigg[\sum_{\nu_{\texttt{A}}}\frac{1}{\nu_{\texttt{A}}}f_{\nu_{\texttt{A}}}(\zeta,\eta,\tau)\bigg]^{\rm cation}_{\rm \texttt{A}^{\nu_{\texttt{A}} +}} + x\bigg[\sum_{\nu_{\texttt{X}}}\frac{1}{\nu_{\texttt{X}}}f_{\nu_{\texttt{X}}}(\zeta,\eta,\tau)\bigg]^{\rm cation}_{\rm \texttt{X}^{\nu_{\texttt{X}} +}} \nonumber \\&& + b\bigg[f_{\nu = 1}(\zeta,\eta,\tau)\bigg]^{\rm anion}_{\rm \texttt{B}^{+}} + y\bigg[f_{\nu = 1}(\zeta,\eta,\tau)\bigg]^{\rm anion}_{\rm \texttt{Y}^{+}} \nonumber \\&=& \bigg[\sum_ql_q\sum_{\nu}\frac{1}{\nu}f_{q,\nu}(\zeta,\eta,\tau)\bigg]_{\rm cation} + \bigg[\sum_ql_qf_{q,\nu = 1}(\zeta,\eta,\tau)\bigg]_{\rm anion}. \label{eq:13}
\end{eqnarray}
Here, $q$ = \texttt{A}, \texttt{B}, \texttt{X} and \texttt{Y}, and therefore, $l_{\texttt{A}} = a$, $l_{\texttt{B}} = b$, $l_{\texttt{X}} = x$ and $l_{\texttt{Y}} = y$. Recall that $\nu < \zeta$ for any cation, and it denotes the number of outer electrons of an electron-donating atom interacting with an electron-accepting atom. Whereas, $\nu$ always equals one for all anions (due to \ref{P16}) where $\nu_{\texttt{B}}$ and $\nu_{\texttt{Y}}$ do not play any role in Eq.~(\ref{eq:13}). 

We are now in the final stage of proving the proposition \ref{P6}$(a)$. In what follows is the proof for $\xi^{\rm system} \propto E_{\rm I}^{\rm system}$ (recall \ref{D7}). The strategy here is to prove $E_{\rm I}^{\rm 
LiCl} > E_{\rm I}^{\rm NaCl} > E_{\rm I}^{\rm KCl} > E_{\rm I}^{\rm RbCl}$ corresponds to $\xi_{\rm LiCl} > \xi_{\rm NaCl} > \xi_{\rm KCl} > \xi_{\rm RbCl}$, and this is sufficient. Of course, one can construct any system from the periodic table such that $E_{\rm I}^{\rm system_1} > E_{\rm I}^{\rm system_2} > E_{\rm I}^{\rm system_3} > E_{\rm I}^{\rm system_4}$ corresponds to $\xi_{\rm LiCl} > \xi_{\rm NaCl} > \xi_{\rm KCl} > \xi_{\rm RbCl}$, but this is a false correspondence. For example, a system consisting of Si and C atoms satisfies $E_{\rm I}^{{\rm Si}_{x_1}{\rm C}_{1-x_1}} > E_{\rm I}^{{\rm Si}_{x_2}{\rm C}_{1-x_2}} > E_{\rm I}^{{\rm Si}_{x_3}{\rm C}_{1-x_3}} > E_{\rm I}^{{\rm Si}_{x_4}{\rm C}_{1-x_4}}$ if $\xi_{\rm C} > \xi_{\rm Si}$ (from Eq.~(\ref{eq:2})), and this nicely corresponds to $\xi_{\rm LiCl} > \xi_{\rm NaCl} > \xi_{\rm KCl} > \xi_{\rm RbCl}$ where $x_i$ denotes the concentration of Si atom, $x_1 < x_2 < x_3 < x_4$, $x_i \in (0,1)$ and $x_i \in \mathbb{R}^+$. Note here that increasing $x_i$ means decreasing concentration of carbon atoms, which in turn implies a decreasing magnitude for $E_{\rm I}^{{\rm Si}_{x_i}{\rm C}_{1-x_i}}$. But this is a false correspondence simply because the chemical elements and their concentrations in the system, Si$_{x_i}$C$_{1-x_i}$ is different by definition from the (Li,Na,K,Rb)Cl system. In fact, this is just one example of the lucky coincidences proven to exist.

The real system, Si$_{x_i}$C$_{1-x_i}$ introduced earlier consists of the chemical elements, Si and C such that $f(\zeta,\eta,\tau) = f_{{\rm Si}_{x_i}{\rm C}_{1-x_i}}(\zeta_{{\rm Si}_{x_i}{\rm C}_{1-x_i}},\eta_{{\rm Si}_{x_i}{\rm C}_{1-x_i}},\tau_{{\rm Si}_{x_i}{\rm C}_{1-x_i}})$. It is also known that each Si has $\zeta = \eta = \tau = 14$, while each C has $\zeta = \eta = \tau = 6$. Therefore, for any $x_i$ defined above, $f(\zeta,\eta,\tau) = f_{{\rm Si}_{x_i}{\rm C}_{1-x_i}}(\zeta_{{\rm Si}_{x_i}{\rm C}_{1-x_i}},\eta_{{\rm Si}_{x_i}{\rm C}_{1-x_i}},\tau_{{\rm Si}_{x_i}{\rm C}_{1-x_i}})$ is unique (from \ref{P4} and Eq.~(\ref{eq:2})) where $\zeta({\rm Si}_{x_i}{\rm C}_{1-x_i}) = \zeta({\rm Si}_{x_i}) + \zeta({\rm C}_{1-x_i}) = 14(x_i) + 6(1-x_i)$, $\eta({\rm Si}_{x_i}{\rm C}_{1-x_i}) = \eta({\rm Si}_{x_i}) + \eta({\rm C}_{1-x_i}) = 14(x_i) + 6(1-x_i)$, and $\tau({\rm Si}_{x_i}{\rm C}_{1-x_i}) = \tau({\rm Si}_{x_i}) + \tau({\rm C}_{1-x_i}) = 14(x_i) + 6(1-x_i)$. Consequently, 
\begin{eqnarray}
f_{{\rm Si}_{x_i}{\rm C}_{1-x_i}}(\zeta_{{\rm Si}_{x_i}{\rm C}_{1-x_i}},\eta_{{\rm Si}_{x_i}{\rm C}_{1-x_i}},\tau_{{\rm Si}_{x_i}{\rm C}_{1-x_i}}) &\propto& x_if_{{\rm Si}}(\zeta_{{\rm Si}},\eta_{{\rm Si}},\tau_{{\rm Si}}) \nonumber \\&& + (1-x_i)f_{{\rm C}}(\zeta_{{\rm C}},\eta_{{\rm C}},\tau_{{\rm C}}), \label{eq:14}
\end{eqnarray}
where the left-hand side is unique because both terms on the right-hand side are themselves unique. Therefore
\begin{eqnarray}
\xi_{{\rm Si}_{x_i}{\rm C}_{1-x_i}} \propto x_iE_{\rm I}^{{\rm Si}} + (1-x_i)E_{\rm I}^{{\rm C}}. \label{eq:15}
\end{eqnarray}
The protons and neutrons in Si$_{x_i}$C$_{1-x_i}$ do not form a single nucleus with its nucleon number equals $\eta_{{\rm Si}_{x_i}{\rm C}_{1-x_i}} + \tau_{{\rm Si}_{x_i}{\rm C}_{1-x_i}} = 2[14(x_i) + 6(1-x_i)]$. Instead, the nuclei of Si and C can arrange to form any crystalline or non-crystalline structure they ``see'' fit due to electron-electron and electron-nucleus interactions. This means that $\xi_{{\rm Si}_{x_i}{\rm C}_{1-x_i}} \rightarrow \xi_{{\rm Si}_{x_i}{\rm C}_{1-x_i}}(\textbf{k})$ where $\textbf{k} = 1/a_L$, $\textbf{k}$ and $a_L$ denote the wavevector and the lattice parameter, respectively. For example, for a crystal structure with repeating cubic unit cells, the lattice parameter $a_L$ defines the size of the unit cell, which is a constant in all three axes. 

\begin{definition} 
\label{D20} 
For $\textbf{k}$-dependent cases, 
\begin{eqnarray}
&&\xi_{{\rm Si}_{x_i}{\rm C}_{1-x_i}}(\textbf{k}) = f_{{\rm Si}_{x_i}{\rm C}_{1-x_i}}(\zeta_{{\rm Si}_{x_i}{\rm C}_{1-x_i}},\eta_{{\rm Si}_{x_i}{\rm C}_{1-x_i}},\tau_{{\rm Si}_{x_i}{\rm C}_{1-x_i}};\textbf{k}), \label{eq:15a}
\end{eqnarray}
implies $\xi_{{\rm Si}_{x_i}{\rm C}_{1-x_i}}$ is now $\textbf{k}$-dependent because $x_i$ is $\textbf{k}$-dependent. Similarly, for time($t$)-dependent cases, one has 
\begin{eqnarray}
&&\xi_{{\rm Si}_{x_i}{\rm C}_{1-x_i}} \rightarrow \xi_{{\rm Si}_{x_i}{\rm C}_{1-x_i}}(t) = f_{{\rm Si}_{x_i}{\rm C}_{1-x_i}}(\zeta_{{\rm Si}_{x_i}{\rm C}_{1-x_i}},\eta_{{\rm Si}_{x_i}{\rm C}_{1-x_i}},\tau_{{\rm Si}_{x_i}{\rm C}_{1-x_i}};t), \label{eq:15b}
\end{eqnarray}
because $x_i \rightarrow x_i(t)$. These correspondence rules can be used depending whether $x_i$ is a constant for a given $i$ or $x_i \rightarrow x_i(\textbf{k})$ or $x_i \rightarrow x_i(t)$ or $x_i \rightarrow x_i(\textbf{k},t) \rightarrow x_i(\textbf{k}(t))$ $\blacksquare$ 
\end{definition} 

Add to that, the exact numbers of electrons, protons and neutrons in a given non-atomic and non-molecular system can be calculated after converting the normalized atomic concentrations, $x_i \in (0,1)$ to chemical units of atomic concentrations per mole. Using Eqs.~(\ref{eq:14}) and~(\ref{eq:15})              
\begin{eqnarray}
\xi_{{\rm (Li,Na,K,Rb)Cl}} \propto E_{I}^{{\rm (Li,Na,K,Rb)Cl}}, \label{eq:16}
\end{eqnarray}
and therefore, $E_{\rm I}^{\rm LiCl} > E_{\rm I}^{\rm NaCl} > E_{\rm I}^{\rm KCl} > E_{\rm I}^{\rm RbCl}$ corresponds to $\xi_{\rm LiCl} > \xi_{\rm NaCl} > \xi_{\rm KCl} > \xi_{\rm RbCl}$.
\end{proof} 
In summary, we have developed a string of precise definitions and logical proofs required to establish what constitutes the ionization energy theory and its approximation. \textit{Nota bene}, all the listed definitions and propositions proven thus far never refer to any experimental observation \textit{a priori}. As a matter of fact, the fundamentals of ionization energy theory actually deal with the questions why and how the energy-level spacings are different, and unique for each chemical element, and also why and how such uniqueness can be extended to any system consisting of these chemical elements by making use of the relation $f(\zeta,\eta,\tau) = \xi$ (see \ref{P4}). However, IET does not provide any prescription on how to calculate $f(\zeta,\eta,\tau)$ or the energy-level spacings ($\xi$) for any chemical element or system. In fact, IET was never designed to do that. The irony here is that, even though ``proving'' IET was a non-trivial task, but applying the ionization energy theory to real systems (biological, chemical or physical) is surprisingly straightforward, unambiguous and can be made quantitative using the so-called ionization energy approximation (or also known as the ionization energy averaging). 

The relevant examples of such applications with quantitative analyses can be found in these reports~\cite{ADA4,ADA7,ADA8,AA1,ADA10,ADA11,ADA12,SYH1,AER1,ADA14,MM1,DHS1,ADA15,ADA16,ADA17,FS1,DR1,KD1,KE1,ADA18,ADA19,ADA20,radha}. Anyway, we are basically done listing the complete and self-consistent definitions, propositions and proofs needed for the ionization energy theory to be considered technically ``correct''. The applications referenced above do suggest the descriptions given by IET are also correct and sufficient, supported by the experiments. However, we cannot be sure whether they (the physical descriptions) are complete, at least, for now. The subsequent sections will provide the necessary setting to further elucidate why IET needs to be invoked so that the physical-sets, -classes and -categories can be constructed such that they have some basic mathematical structures. 

\section{Energy-level spacing physical categories and their constructions}

According to Geroch~\cite{RG1}, we need three things to properly define a mathematical category ($\mathscr{M}$). An abstract definition is given in \ref{D21} [see page 3 in Ref.~\cite{RG1}].

\begin{definition}
\label{D21}
$(i)$ A class, $\mathcal{C}$ has objects as their elements, and these objects ($\Omega_{\rm A}, \Omega_{\rm B}, \cdots$) can also be regarded as sets. \\
$(ii)$ A set denoted by $\texttt{Mor}(\Omega_{\rm A},\Omega_{\rm B})$ has elements, which are called morphisms from $\Omega_{\rm A}$ to $\Omega_{\rm B}$.\\
$(iii)$ Any morphism, $\Phi_{\Omega_{\rm A},\Omega_{\rm B}}$ is from $\Omega_{\rm A}$ to $\Omega_{\rm B}$, and any morphism $\Phi_{\Omega_{\rm B},\Omega_{\rm C}}$ is from $\Omega_{\rm B}$ to $\Omega_{\rm C}$. These morphisms can be written as a composition such that $\Phi_{\Omega_{\rm B},\Omega_{\rm C}} \circ \Phi_{\Omega_{\rm A},\Omega_{\rm B}}$ is from $\Omega_{\rm A}$ to $\Omega_{\rm C}$. The composition of $\Phi_{\Omega_{\rm B},\Omega_{\rm C}}$ with $\Phi_{\Omega_{\rm A},\Omega_{\rm B}}$ need to satisfy two additional conditions listed below.\\
$(a)$ Associativity--- suppose there are four objects, $\Omega_{\rm A}, \Omega_{\rm B}, \Omega_{\rm C}$ and $\Omega_{\rm D}$, and $\Phi_{\Omega_{\rm A},\Omega_{\rm B}}$, $\Phi_{\Omega_{\rm B},\Omega_{\rm C}}$ and $\Phi_{\Omega_{\rm C},\Omega_{\rm D}}$ are the morphisms, then 
\begin{eqnarray}
(\Phi_{\Omega_{\rm C},\Omega_{\rm D}} \circ \Phi_{\Omega_{\rm B},\Omega_{\rm C}}) \circ \Phi_{\Omega_{\rm A},\Omega_{\rm B}} = \Phi_{\Omega_{\rm C},\Omega_{\rm D}} \circ (\Phi_{\Omega_{\rm B},\Omega_{\rm C}} \circ \Phi_{\Omega_{\rm A},\Omega_{\rm B}}). \label{eq:17}
\end{eqnarray}
$(b)$ Identities exist--- for each object, $\Omega_{\rm A}, \Omega_{\rm B}, \cdots$, there exists an identity morphism, $\Phi^{\rm identity}_{\Omega_{\rm A},\Omega_{\rm A}}$ from $\Omega_{\rm A}$ to $\Omega_{\rm A}$ such that 
\begin{eqnarray}
&&\Phi_{\Omega_{\rm A},\Omega_{\rm B}} \circ \Phi^{\rm identity}_{\Omega_{\rm A},\Omega_{\rm A}} = \Phi_{\Omega_{\rm A},\Omega_{\rm B}}~~~{\rm and} \label{eq:18} \\&&
\Phi^{\rm identity}_{\Omega_{\rm A},\Omega_{\rm A}} \circ \Phi_{\Omega_{\rm B},\Omega_{\rm A}} = \Phi_{\Omega_{\rm B},\Omega_{\rm A}}. \label{eq:19}
\end{eqnarray}

In addition, $(ii)$ and $(iii)$ also imply that the objects (stated above) are ordered, and if there are two elements in a given object, then these elements form an ordered pair. This ordering is inapplicable for singletons $\blacksquare$
\end{definition}

\begin{proposition}
\label{P21}
All energy-level spacing physical categories $(^{\xi}\mathscr{P})$ satisfy $\ref{D21}$. 
\end{proposition}

\begin{proof}
This proof consists of three parts--- the first is for atoms and ions, the second is for molecules and molecular ions, while the last part is for compounds (excluding free-electron metals, see \ref{D18}).\\ \\
\texttt{Part I}: Atoms ($\chi_{i,j}$) and ions ($\chi_{i,j}^{\omega'+}$, $\chi_{i,j}^{\alpha-,\alpha'-}$)\\
We first construct the singletons from isolated atoms and ions by recalling \ref{P4}, \ref{D8} \ref{P11} and \ref{P12} such that $\zeta \geq 1$ (there is at least one bounded electron), and if $\zeta > \tau$ then $\alpha > 0$ giving rise to anions. For example, $\{\rm H^{\alpha'-}\}$, $\cdots$, $\{\rm H^{\alpha-}\}$, $\{\rm H\}$, $\{\rm He^{\alpha'-}\}$, $\cdots$, $\{\rm He^{\alpha-}\}$, $\{\rm He\}$, $\{\rm He^{+}\}$, $\{\rm Li^{\alpha'-}\}$, $\cdots$, $\{\rm Li^{\alpha-}\}$, $\{\rm Li\}$, $\{\rm Li^{+}\}$, $\{\rm Li^{2+}\}$, $\{\rm Be^{\alpha'-}\}$, $\cdots$, $\{\rm Be^{\alpha-}\}$, $\{\rm Be\}$, $\{\rm Be^{+}\}$, $\{\rm Be^{2+}\}$, $\{\rm Be^{3+}\}$, $\{\chi_i^{\alpha'-}\}$, $\cdots$, $\{\chi_i^{\alpha-}\}$, $\{\chi_i\}$, $\cdots$, $\{\chi_i^{\omega'+}\}$, $\cdots$, $\{\chi_j^{\alpha'-}\}$, $\cdots$, $\{\chi_j^{\alpha-}\}$, $\{\chi_j\}$, $\cdots$, $\{\chi_j^{\omega'+}\}$. Following \ref{P10}, we have to exclude these singletons, $\{\rm H^{+}\}$, $\{\rm He^{2+}\}$, $\{\rm Li^{3+}\}$, $\{\rm Be^{4+}\}$, $\cdots$, $\{\chi_i^{Z_i+}\}$, $\cdots$, $\{\chi_j^{Z_j+}\}$ because each and every one of them are electronless cations, and therefore $f(\zeta_{= 0},\eta_{\geq 0},\tau_{\geq 0}) = \xi = 0$ where $\alpha < \alpha'$, $\{i,j,\alpha,\alpha',\omega'\} \in \mathbb{N}^*$, $j > i$ and $Z_j > Z_i$ such that $\omega_i' \neq Z_i$ and $\omega_i' < Z_i$. 

On the other hand, the singletons with one or more bounded electrons ($\zeta_i \geq 1$) form the energy-level spacing set (also known as the object from \ref{D21}), and all of them (singletons) are members of the class of isolated atoms and ions, $\mathcal{C}^{\rm atom}_{\rm ion}$ where $\mathcal{C}^{\rm atom}_{\rm ion} \in$ $^{\xi}\mathscr{P}^{\rm atom}_{\rm ion}$, $^{\xi}\mathscr{P}^{\rm atom}_{\rm ion}$ is the atomic-ionic energy-level spacing physical category. Of course the sets of electronless cations can also be collected to form a class such that $\mathcal{C}_{\rm nucleus} \in$ $\mathscr{P}_{\rm nucleus}$ where $f(\zeta_{= 0},\eta_{\geq 0},\tau_{\geq 1}) \rightarrow f_{\rm nucleus}(\eta_{\geq 0},\tau_{\geq 1}) \neq 0$. Unlike $f(\zeta_{\geq 1},\eta_{\geq 0},\tau_{\geq 1})$, which has been logically proven to be equal to $\xi$ (see \ref{P4} and \ref{P10}), we have not found a suitable physical quantity that can also be logically associated to $f_{\rm nucleus}(\eta_{\geq 0},\tau_{\geq 1})$. Moreover, $f_{\rm nucleus}(\eta_{\geq 0},\tau_{\geq 1})$ cannot be associated to $\xi$ (because $\zeta = 0$), and therefore, $\mathscr{P}_{\rm nucleus}$ is a nucleus physical category, not an energy-level spacing physical category. 

Interestingly, each energy-level spacing physical category consists of all three things stated in \ref{D21}--- $^{\xi}\mathscr{P}^{\rm atom}_{\rm ion}$ has a class, $\mathcal{C}^{\rm atom}_{\rm ion}$ that has singletons as their elements ($\chi_{i,j}^{\alpha'-}$, $\chi_{i,j}^{\alpha-}$, $\chi_{i,j}$, $\chi_{i,j}^{\omega'+}$), such that $\{\chi_{i,j}^{\alpha'-}\}, \cdots, \{\chi_{i,j}^{\alpha-}\}, \{\chi_{i,j}\}, \cdots, \{\chi_{i,j}^{\omega'+}\} \in$ $\mathcal{C}^{\rm atom}_{\rm ion} \in$ $^{\xi}\mathscr{P}^{\rm atom}_{\rm ion}$. Each singleton in each of these sets, $\{\chi_i^{\alpha'-}\}$, $\cdots$, $\{\chi_i^{\alpha-}\}$, $\{\chi_i\}$, $\cdots$, $\{\chi_i^{\omega'+}\}$, $\cdots$ $\{\chi_j^{\alpha'-}\}$, $\cdots$, $\{\chi_j^{\alpha-}\}$, $\{\chi_j\}$, $\cdots$, $\{\chi_j^{\omega'+}\}$ is a physical entity consisting of at least one nucleus, with at least one bounded electron. Inside this nucleus, there must be at least one proton or a collection of any number of protons and neutrons. 

The second thing that we require is the existence of a set denoted by $\texttt{Mor}(\Omega_{\rm A},\Omega_{\rm B})$. In fact, we can obtain such morphisms by noting the first four ordered objects, which are the elements of $\mathcal{C}^{\rm atom}_{\rm ion}$ where $\Omega_{\rm A} = \{\rm H\}, \Omega_{\rm B} = \{\rm He\}, \Omega_{\rm C} = \{\rm He^{+}\}$, and $\Omega_{\rm D} = \{\rm Li\}$, and therefore, $\Phi_{\{\rm H\},\{\rm He\}}$ is the morphism from $\{\rm H\}$ to $\{\rm He\}$, $\Phi_{\{\rm He\},\{\rm He^{+}\}}$ is the morphism from $\{\rm He\}$ to $\{\rm He^{+}\}$ and $\Phi_{\{\rm He^{+}\},\{\rm Li\}}$ is the morphism from $\{\rm He^{+}\}$ to $\{\rm Li\}$ where $\Phi_{\{\rm H\},\{\rm He\}} \in \texttt{Mor}(\{\rm H\},\{\rm He\})$, $\Phi_{\{\rm He\},\{\rm He^{+}\}} \in \texttt{Mor}(\{\rm He\},\{\rm He^+\})$ and $\Phi_{\{\rm He^{+}\},\{\rm Li\}} \in \texttt{Mor}(\{\rm He^+\},\{\rm Li\})$. We can now make use of $f(\zeta,\eta,\tau)$ so that $\Phi_{\{\rm H\},\{\rm He\}} = f(\zeta_{\rm H}+1,\eta_{\rm H}+2,\tau_{\rm H}+1)$, $\Phi_{\{\rm He\},\{\rm He^+\}} = f(\zeta_{\rm He}-1,\eta_{\rm He}+0,\tau_{\rm He}+0)$, $\Phi_{\{\rm He^+\},\{\rm Li\}} = f(\zeta_{\rm He^+}+2,\eta_{\rm He^+}+3,\tau_{\rm He^+}+1)$, and these morphisms satisfy the associativity condition,
\begin{eqnarray}
\big(\Phi_{\{\rm He^{+}\},\{\rm Li\}} \circ \Phi_{\{\rm He\},\{\rm He^{+}\}}\big) \circ \Phi_{\{\rm H\},\{\rm He\}} &=& \Phi_{\{\rm H\},\{\rm Li\}} \nonumber \\&=& f(\zeta_{\rm H}+2,\eta_{\rm H}+5,\tau_{\rm H}+2),
\label{eq:20}
\end{eqnarray}
and
\begin{eqnarray}
\Phi_{\{\rm He^{+}\},\{\rm Li\}} \circ \big(\Phi_{\{\rm He\},\{\rm He^{+}\}} \circ \Phi_{\{\rm H\},\{\rm He\}}\big) &=& \Phi_{\{\rm H\},\{\rm Li\}} \nonumber \\&=& f(\zeta_{\rm H}+2,\eta_{\rm H}+5,\tau_{\rm H}+2), \label{eq:21}
\end{eqnarray}
and therefore
\begin{eqnarray}
(\Phi_{\{\rm He^{+}\},\{\rm Li\}} \circ \Phi_{\{\rm He\},\{\rm He^{+}\}}) \circ \Phi_{\{\rm H\},\{\rm He\}} = \Phi_{\{\rm He^{+}\},\{\rm Li\}} \circ (\Phi_{\{\rm He\},\{\rm He^{+}\}} \circ \Phi_{\{\rm H\},\{\rm He\}}). \nonumber \\&& \label{eq:22}
\end{eqnarray}
Here, for example, $f(\zeta_{\rm H}+1,\eta_{\rm H}+2,\tau_{\rm H}+1)$ physically involves a nuclear reaction such that $\eta_{\rm H}+2$ and $\tau_{\rm H}+1$ are obtained, and followed by electron-addition, $\zeta_{\rm H}+1$ (chemical reaction) to produce a new chemical element, He. In addition, both $\Phi_{\{\rm H\},\{\rm He\}}$ and $f(\zeta_{\rm H}+1,\eta_{\rm H}+2,\tau_{\rm H}+1)$ physically imply changing interaction strengths with respect to electron-electron and electron-nucleus interactions, and/or creation of new spin-spin interaction between electrons, and between electrons and the nucleus. This explains why IET is suitable to evaluate the changing interaction strengths in any quantum matter, excluding the free-electron metals. 

Finally, we show the existence of identities, $\Phi^{\rm identity}_{\{\rm H\},\{\rm H\}} = f(\zeta_{\rm H}+0,\eta_{\rm H}+0,\tau_{\rm H}+0)$ maps $\{\rm H\}$ to $\{\rm H\}$, $\Phi^{\rm identity}_{{\{\rm He}\},\{\rm He\}} = f(\zeta_{\rm He}+0,\eta_{\rm He}+0,\tau_{\rm He}+0)$ maps $\{\rm He\}$ to $\{\rm He\}$, and so on. These identities satisfy
\begin{eqnarray}
\Phi_{\{\rm H\},\{\rm He\}} \circ \Phi^{\rm identity}_{\{\rm H\},\{\rm H\}} = \Phi_{\{\rm H\},\{\rm He\}},
\label{eq:23}
\end{eqnarray}
and
\begin{eqnarray}
\Phi^{\rm identity}_{\{\rm He\},\{\rm He\}} \circ \Phi_{\{\rm H\},\{\rm He\}} = \Phi_{\{\rm H\},\{\rm He\}}. \label{eq:24}
\end{eqnarray}
Note that both associativity and identities exist because each mapping requires $\zeta + \mathbb{Z}, \eta + \mathbb{Z}$ and $\tau + \mathbb{Z}$ where $\mathbb{Z}$ is the set of integers.\\
\texttt{WARNING}:
\begin{eqnarray}
f(\zeta_{\rm He}-1,\eta_{\rm He}-2,\tau_{\rm He}-1) &\neq& f^{-1}(\zeta_{\rm H},\eta_{\rm H},\tau_{\rm H}) \nonumber \\&\neq& f^{-1}(\zeta_{\rm He},\eta_{\rm He},\tau_{\rm He}) \nonumber \\&=& f(\zeta_{\rm H},\eta_{\rm H},\tau_{\rm H}).
\label{eq:23warn}
\end{eqnarray} \\
\texttt{Part II}: Molecules ($\gamma_j$) and molecular ions ($\gamma_j^{\alpha'-}$, $\gamma_j^{\omega'+}$)\\
Each molecule and its molecular ion forms an energy-level spacing category, for example, $^{\xi}\mathscr{P}_{\rm H_2O}$, $^{\xi}\mathscr{P}_{\rm H_2}$, and so on such that the only member of the category, $^{\xi}\mathscr{P}_{\rm H_2}$ is the class, $\mathcal{C}_{\rm H_2}$ that has singletons as members, namely, $\{\rm H_2^{\alpha'-}\}$, $\cdots$, $\{\rm H_2\}$, $\cdots$, $\{\rm H_2^{\omega'+}\}$. Each singleton is a molecule or its molecular ion that can be denoted generally by $\gamma_j$ such that $\{\gamma_j^{\alpha'-}\}$, $\cdots$, $\{\gamma_j\}$, $\cdots$, $\{\gamma_j^{\omega'+}\} \in$ $\mathcal{C}_{\gamma_j} \in$ $^{\xi}\mathscr{P}_{\gamma_j}$. In other words, each $^{\xi}\mathscr{P}_{\gamma_j}$ represents only one type of molecule and its molecular ion, for example $\mathcal{C}_{\rm H_2O} \notin$ $^{\xi}\mathscr{P}_{\rm H_2}$. Hence, we have established the existence of a class. Next, we need to find the morphisms between these singletons where 
\begin{eqnarray}
&&\Omega_{\rm A} = \{\rm (H_2O)^{2-}\}, \nonumber \\&& \Omega_{\rm B} = \{\rm (H_2O)^{-}\}, \nonumber \\&& \Omega_{\rm C} = \{\rm H_2O\}, \nonumber \\&& \Omega_{\rm D} = \{\rm (H_2O)^{+}\}, \label{eq:24new1}
\end{eqnarray}
and therefore, $\Phi_{\{\rm (H_2O)^{2-}\},\{\rm (H_2O)^{-}\}}$ is the morphism from $\{\rm (H_2O)^{2-}\}$ to $\{\rm (H_2O)^{-}\}$, $\Phi_{\{\rm (H_2O)^{-}\},\{\rm H_2O\}}$ is the morphism from $\{\rm (H_2O)^{-}\}$ to $\{\rm H_2O\}$ and $\Phi_{\{\rm H_2O\},\{\rm (H_2O)^{+}\}}$ is the morphism from $\{\rm H_2O\}$ to $\{\rm (H_2O)^{+}\}$ where 
\begin{eqnarray}
&&\Phi_{\{\rm (H_2O)^{2-}\},\{\rm (H_2O)^{-}\}} \in \texttt{Mor}(\{\rm (H_2O)^{2-}\},\{\rm (H_2O)^{-}\}), \nonumber \\&& \Phi_{\{\rm (H_2O)^{-}\},\{\rm H_2O\}} \in \texttt{Mor}(\{\rm (H_2O)^{-}\},\{\rm H_2O\}), \nonumber \\&& \Phi_{\{\rm H_2O\},\{\rm (H_2O)^{+}\}} \in \texttt{Mor}(\{\rm H_2O\},\{\rm (H_2O)^{+}\}). \label{eq:24new2}
\end{eqnarray}
Similar to atoms and ions in Part I, we again make use of $f(\zeta,\eta,\tau)$ such that $\Phi_{\{\rm (H_2O)^{2-}\},\{\rm (H_2O)^{-}\}} = f_{\rm (H_2O)^{2-}}(\zeta -1,\eta +0,\tau +0)$, $\Phi_{\{\rm (H_2O)^{-}\},\{\rm H_2O\}} = f_{\rm (H_2O)^{-}}(\zeta -1,\eta +0,\tau +0)$, $\Phi_{\{\rm H_2O\},\{\rm (H_2O)^+\}} = f_{\rm H_2O}(\zeta -1,\eta +0,\tau +0)$, and these morphisms do satisfy the associativity condition,    
\begin{eqnarray}
&&\big(\Phi_{\{\rm H_2O\},\{\rm (H_2O)^{+}\}} \circ \Phi_{\{\rm (H_2O)^{-}\},\{\rm H_2O\}}\big) \circ \Phi_{\{\rm (H_2O)^{2-}\},\{\rm (H_2O)^{-}\}} \nonumber \\&& = \Phi_{\{\rm (H_2O)^{2-}\},\{\rm (H_2O)^{+}\}} = f_{\rm (H_2O)^{2-}}(\zeta -3,\eta +0,\tau +0),
\label{eq:25}
\end{eqnarray}
and
\begin{eqnarray}
&&\Phi_{\{\rm H_2O\},\{\rm (H_2O)^{+}\}} \circ \big(\Phi_{\{\rm (H_2O)^{-}\},\{\rm H_2O\}} \circ \Phi_{\{\rm (H_2O)^{2-}\},\{\rm (H_2O)^{-}\}}\big) \nonumber \\&& = \Phi_{\{\rm (H_2O)^{2-}\},\{\rm (H_2O)^{+}\}} = f_{\rm (H_2O)^{2-}}(\zeta -3,\eta +0,\tau +0), \label{eq:26}
\end{eqnarray}
and therefore Eq.~(\ref{eq:25}) = Eq.~(\ref{eq:26}) satisfying Eq.~(\ref{eq:17}). Identities do exist as required, in particular, 
\begin{eqnarray}
\Phi^{\rm identity}_{\rm \{(H_2O)^{2-}\},\{(H_2O)^{2-}\}} = f_{\rm (H_2O)^{2-}}(\zeta +0,\eta +0,\tau +0), \label{eq:26new1}
\end{eqnarray}
maps $\{\rm (H_2O)^{2-}\}$ to $\{\rm (H_2O)^{2-}\}$, and 
\begin{eqnarray}
\Phi^{\rm identity}_{\rm \{(H_2O)^{-}\},\{(H_2O)^{-}\}} = f_{\rm (H_2O)^{-}}(\zeta +0,\eta +0,\tau +0), \label{eq:26new2}
\end{eqnarray}
maps $\{\rm (H_2O)^{-}\}$ to $\{\rm (H_2O)^{-}\}$ such that
\begin{eqnarray}
&&\Phi_{\{\rm (H_2O)^{2-}\},\{\rm (H_2O)^{-}\}} \circ \Phi^{\rm identity}_{\{\rm (H_2O)^{2-}\},\{\rm (H_2O)^{2-}\}} = \Phi_{\{\rm (H_2O)^{2-}\},\{\rm (H_2O)^{-}\}},
\label{eq:28}
\end{eqnarray}
and
\begin{eqnarray}
&&\Phi^{\rm identity}_{\{\rm (H_2O)^{-}\},\{\rm (H_2O)^{-}\}} \circ \Phi_{\{\rm (H_2O)^{2-}\},\{\rm (H_2O)^{-}\}} = \Phi_{\{\rm (H_2O)^{2-}\},\{\rm (H_2O)^{-}\}}. \label{eq:29}
\end{eqnarray}
We now proceed to the last part. \\ \\
\texttt{Part III}: Compounds ($\beta_j$)\\
We now recall the compound, Si$_{x_i}$C$_{1-x_i}$, which forms the category, $^{\xi}\mathscr{P}_{{\rm Si}_{x_i}{\rm C}_{1-x_i}}$, and each different compound, $\beta_j$ forms its own category, $^{\xi}\mathscr{P}_{\beta_j}$ where $\beta_1 =$ Si$_{x_i}$C$_{1-x_i}$ is just one of them and $\beta_2$ is another compound, other than the Si-C combination. The class, $\mathcal{C}_{\beta_j} \in$ $^{\xi}\mathscr{P}_{\beta_j}$, $\mathcal{C}_{\beta_j} \notin$ $^{\xi}\mathscr{P}_{\beta_i}$ and as usual, each class has singletons as members, namely, $\{{\rm Si}_{x_1}{\rm C}_{1-x_1}\}$, $\cdots$, $\{{\rm Si}_{x_2}{\rm C}_{1-x_2}\}$, $\cdots$, where $\{i,j\} \in \mathbb{N}^*$, $\{x_i\} \in (0,1) \in \mathbb{R}^+$, and $\{{\rm Si}_{x_i}{\rm C}_{1-x_i}\}$ is a singleton for each $i$. Subsequently, similar to Parts I and II, morphisms can be easily defined to satisfy the associativity condition (see Eq.~(\ref{eq:17}))   
\begin{eqnarray}
&&\big(\Phi_{\{{\rm Si}_{x_2}{\rm C}_{1-x_2}\},\{{\rm Si}_{x_1}{\rm C}_{1-x_1}\}} \circ \Phi_{\{{\rm Si}_{x_3}{\rm C}_{1-x_3}\},\{{\rm Si}_{x_2}{\rm C}_{1-x_2}\}}\big) \circ \Phi_{\{{\rm Si}_{x_4}{\rm C}_{1-x_4}\},\{{\rm Si}_{x_3}{\rm C}_{1-x_3}\}} \nonumber \\&& = \Phi_{\{{\rm Si}_{x_4}{\rm C}_{1-x_4}\},\{{\rm Si}_{x_1}{\rm C}_{1-x_1}\}} \nonumber \\&& = f_{{\rm Si}_{x_1}{\rm C}_{1-x_1}}\big(\big[\zeta_{{\rm Si}_{x_1}{\rm C}_{1-x_1}} + (x_4 - x_1)\zeta_{\rm Si} + (x_1 - x_4)\zeta_{\rm C}\big], \big[\eta_{{\rm Si}_{x_1}{\rm C}_{1-x_1}} \nonumber \\&& + (x_4 - x_1)\eta_{\rm Si} + (x_1 - x_4)\eta_{\rm C}\big], \big[\tau_{{\rm Si}_{x_1}{\rm C}_{1-x_1}} + (x_4 - x_1)\tau_{\rm Si} + (x_1 - x_4)\tau_{\rm C}\big]\big), \nonumber \\ \label{eq:30}
\end{eqnarray}
and
\begin{eqnarray}
&&\Phi_{\{{\rm Si}_{x_2}{\rm C}_{1-x_2}\},\{{\rm Si}_{x_1}{\rm C}_{1-x_1}\}} \circ \big(\Phi_{\{{\rm Si}_{x_3}{\rm C}_{1-x_3}\},\{{\rm Si}_{x_2}{\rm C}_{1-x_2}\}} \circ \Phi_{\{{\rm Si}_{x_4}{\rm C}_{1-x_4}\},\{{\rm Si}_{x_3}{\rm C}_{1-x_3}\}}\big) \nonumber \\&& = \Phi_{\{{\rm Si}_{x_4}{\rm C}_{1-x_4}\},\{{\rm Si}_{x_1}{\rm C}_{1-x_1}\}} \nonumber \\&& = f_{{\rm Si}_{x_1}{\rm C}_{1-x_1}}\big(\big[\zeta_{{\rm Si}_{x_1}{\rm C}_{1-x_1}} + (x_4 - x_1)\zeta_{\rm Si} + (x_1 - x_4)\zeta_{\rm C}\big], \big[\eta_{{\rm Si}_{x_1}{\rm C}_{1-x_1}} \nonumber \\&& + (x_4 - x_1)\eta_{\rm Si} + (x_1 - x_4)\eta_{\rm C}\big], \big[\tau_{{\rm Si}_{x_1}{\rm C}_{1-x_1}} + (x_4 - x_1)\tau_{\rm Si} + (x_1 - x_4)\tau_{\rm C}\big]\big), \nonumber \\ \label{eq:31}
\end{eqnarray}
where Eq.~(\ref{eq:30}) = Eq.~(\ref{eq:31}). It is also straightforward for one to obtain
\begin{eqnarray}
&&\Phi_{\{{\rm Si}_{x_1}{\rm C}_{1-x_1}\},\{{\rm Si}_{x_2}{\rm C}_{1-x_2}\}} \circ \Phi^{\rm identity}_{\{{\rm Si}_{x_1}{\rm C}_{1-x_1}\},\{{\rm Si}_{x_1}{\rm C}_{1-x_1}\}} = \Phi_{\{{\rm Si}_{x_1}{\rm C}_{1-x_1}\},\{{\rm Si}_{x_2}{\rm C}_{1-x_2}\}}, \nonumber \\&& \label{eq:32}
\end{eqnarray}
and
\begin{eqnarray}
&&\Phi^{\rm identity}_{\{{\rm Si}_{x_2}{\rm C}_{1-x_2}\},\{{\rm Si}_{x_2}{\rm C}_{1-x_2}\}} \circ \Phi_{\{{\rm Si}_{x_1}{\rm C}_{1-x_1}\},\{{\rm Si}_{x_2}{\rm C}_{1-x_2}\}} = \Phi_{\{{\rm Si}_{x_1}{\rm C}_{1-x_1}\},\{{\rm Si}_{x_2}{\rm C}_{1-x_2}\}}. \nonumber \\&& \label{eq:33}
\end{eqnarray}
This completes the proof for \ref{P21}. \end{proof}
      
\begin{proposition}
\label{P22}
$(a)$ Inverses, closure and identities do not exist for $\mathcal{C}_{\chi} \in$ $^{\xi}\mathscr{P}^{\mathbb{Z}}_{\chi}$ and $\mathcal{C}_{\gamma} \in$ $^{\xi}\mathscr{P}^{\mathbb{Z}}_{\gamma}$. $(b)$ Inverses and identities do not exist for $\mathcal{C}_{\beta} \in$ $^{\xi}\mathscr{P}^{\mathbb{R}}_{\beta}$, and $(c)$ in view of $\ref{P21}$--- $\mathcal{C}_{\chi}$ and $\mathcal{C}_{\gamma}$ form semicategories while $\mathcal{C}_{\beta}$ forms a semigroup. 
\end{proposition}
\begin{proof}
We have introduced some compact notations, $^{\xi}\mathscr{P}^{\mathbb{Z}}_{\chi}$, $^{\xi}\mathscr{P}^{\mathbb{Z}}_{\gamma}$ and $^{\xi}\mathscr{P}^{\mathbb{R}}_{\beta}$ where $\chi$ here refers to isolated atoms and ions, $\gamma$ denotes molecules and molecular ions, while $\beta$ represents any non-free-electron compound. Moreover, the superscript, $\mathbb{Z}$ implies one needs integers to map one element to another element where both elements are members of their respective physical sets, and these physical sets are members of either $\mathcal{C}_{\chi} \in$ $^{\xi}\mathscr{P}^{\mathbb{Z}}_{\chi}$ or $\mathcal{C}_{\gamma} \in$ $^{\xi}\mathscr{P}^{\mathbb{Z}}_{\gamma}$. Whereas, the superscript $\mathbb{R}$ implies one requires real numbers for mappings between elements contained in $\mathcal{C}_{\beta} \in$ $^{\xi}\mathscr{P}^{\mathbb{R}}_{\beta}$.\\
$(a)$ For every singleton, $\chi \in \mathcal{C}_{\chi} \in$ $^{\xi}\mathscr{P}^{\mathbb{Z}}_{\chi}$ and for every singleton, $\gamma \in \mathcal{C}_{\gamma} \in$ $^{\xi}\mathscr{P}^{\mathbb{Z}}_{\gamma}$, there is not a single inverse, $\chi^{-1} \in \mathcal{C}_{\chi} \in$ $^{\xi}\mathscr{P}^{\mathbb{Z}}_{\chi} \wedge \gamma^{-1} \in \mathcal{C}_{\gamma} \in$ $^{\xi}\mathscr{P}^{\mathbb{Z}}_{\gamma}$ exist such that  
\begin{eqnarray}
&&\chi\chi^{-1} = f(\zeta_{\chi},\eta_{\chi},\tau_{\chi}) f^{-1}(-\zeta_{\chi},-\eta_{\chi},-\tau_{\chi}) = f_{\chi}(0,0,0), \label{eq:34} \\&&
\gamma\gamma^{-1} = f(\zeta_{\gamma},\eta_{\gamma},\tau_{\gamma}) f^{-1}(-\zeta_{\gamma},-\eta_{\gamma},-\tau_{\gamma}) = f_{\gamma}(0,0,0), \label{eq:35}
\end{eqnarray}
because there is no such thing as molecules and/or atoms and/or ions and/or molecular ions with negative amount of particles (electrons, neutrons and protons) where $\{\eta_{\chi,\gamma}\} \in \mathbb{N}$, $\{\tau_{\chi,\gamma},\zeta_{\chi,\gamma}\} \in \mathbb{N}^*$ where $\mathbb{N}$ and $\mathbb{N}^*$ are the sets of natural numbers, including and excluding zero, respectively. Therefore, $f^{-1}(-\zeta_{\chi},-\eta_{\chi},-\tau_{\chi})$ and $f^{-1}(-\zeta_{\gamma},-\eta_{\gamma},-\tau_{\gamma})$ do not exist by definition, which imply the respective inverses, $\chi^{-1}$ and $\gamma^{-1}$ cannot exist either. The condition for closure is also not satisfied. In particular, a new chemical element or ion can never be formed for any $\zeta_{\chi} > 0$ and $\tau_{\chi} > 1$ if $\eta_{\chi} = 0$, and/or for any $\zeta_{\chi} \gg 1$ if $\tau_{\chi} = 1$ and $\eta_{\chi} = 0$. Likewise, for any molecule, $\gamma$, any molecular ion, $\gamma^{\alpha'-}$ cannot exist for any $\alpha'\gg \zeta_{\gamma}$ where $\alpha' \in$ $\mathbb{N}^*$. Identity elements, namely, $f_{\chi}(0,0,0)$ and $f_{\gamma}(0,0,0)$ for molecules, atoms, ions and molecular ions do not exist by definition. \\ \\
\texttt{REMINDER}: Any singleton that carries a specific name for any atom ($\chi_{i,j}$), ion ($\chi_{i,j}^{\omega'+}$, $\chi_{i,j}^{\alpha-,\alpha'-}$), molecule ($\gamma_j$), molecular ion ($\gamma_j^{\alpha'-}$, $\gamma_j^{\omega'+}$) or compound ($\beta_j$) is in itself meaningless and/or ``dead''. Unique energy-level spacings for each singleton give ``life'' (mathematical structures) to every energy-level spacing physical class, and therefore to every energy-level spacing physical category ($^{\xi}\mathscr{P}$) by means of $\xi = f(\zeta,\eta,\tau)$. \\
$(b)$ For any singleton, $\beta \in \mathcal{C}_{\beta} \in$ $^{\xi}\mathscr{P}^{\mathbb{R}}_{\beta}$, an inverse $\beta^{-1} \in \mathcal{C}_{\beta} \in$ $^{\xi}\mathscr{P}^{\mathbb{R}}_{\beta}$ does not exist such that  
\begin{eqnarray}
&&\beta\beta^{-1} = f_{\beta}\big(\big[x_1\zeta +x_2\zeta \cdots +x_i\zeta\big],\big[x_1\eta +x_2\eta \cdots +x_i\eta\big],\big[x_1\tau +x_2\tau \cdots +x_i\tau\big]\big) \nonumber \\&& f_{\beta}^{-1}\big(\big[-x_1\zeta -x_2\zeta \cdots -x_i\zeta\big],\big[-x_1\eta -x_2\eta \cdots -x_i\eta\big],\big[-x_1\tau -x_2\tau \cdots -x_i\tau\big]\big) \nonumber \\&& = f_{\beta}(0,0,0), \nonumber \\&& \label{eq:36}
\end{eqnarray}
because of the same reason stated above--- any compound with negative amount of particles cannot exist where $\{x_1, x_2, \cdots, x_i\} \in (0,1) \in \mathbb{R}$. This means that $f_{\beta}^{-1}\big(\big[-x_1\zeta -x_2\zeta \cdots -x_i\zeta\big],\big[-x_1\eta -x_2\eta \cdots -x_i\eta\big],\big[-x_1\tau -x_2\tau \cdots -x_i\tau\big]\big)$ does not exist. Now, it is straightforward to observe the existence of closure due to $\{x_1, x_2, \cdots, x_i\} \in (0,1) \in \mathbb{R}$ for a given compound. Again, identity elements, namely, $f_{\beta}(0,0,0)$ does not exist by definition for every compound and for every $x_i$ where $\{x_i\} \in (0,1)$. \\
$(c)$ From Parts I and II given in \ref{P21}, and from Eqs.~(\ref{eq:34}) and~(\ref{eq:35}), indeed $\mathcal{C}_{\chi} \in$ $^{\xi}\mathscr{P}^{\mathbb{Z}}_{\chi}$ and $\mathcal{C}_{\gamma} \in$ $^{\xi}\mathscr{P}^{\mathbb{Z}}_{\gamma}$ form semicategories in the absence of identities, inverses and closure. On the other hand, Part III in \ref{P21} and Eq.~(\ref{eq:36}) imply $\mathcal{C}_{\beta}$ $\in$ $^{\xi}\mathscr{P}^{\mathbb{R}}_{\beta}$ is a semigroup without identities and inverses.  
\end{proof}

In any case, you should be aware that there exist many other physical categories other than the energy-level spacing physical categories. One such example has been introduced earlier, which is the $\mathscr{P}_{\rm nucleus}$, and there are also other physical systems completely independent of $^{\xi}\mathscr{P}$, namely, a confined vacuum or any physical space containing only electromagnetic waves (photons), or gravitational fields, or unbounded (free) electrons, or any particles without any bounded electron, or any yet to be discovered particles (or sub-particles) or strings or any combination of them. These physical systems do not belong to the energy-level spacing physical categories constructed herein. The reason is that any physical system in the absence of any ``boundedness'' between a given nucleus and an electron has zero energy-level spacing. 

For example, the nucleus must contain at least one proton with at least one bounded electron, and in this case, this minimal system is an atomic hydrogen with distinct and discrete energy levels, which give rise to the well-defined energy-level spacings. In addition, the class of free-electron quantum matter (that has zero energy-level spacing) forms yet another physical category, $\mathscr{P}^{\rm free}_{\rm electron}$, which is independent of $^{\xi}\mathscr{P}$. All these ``other'' physical systems can exploit the knowledge from the available mathematical categories and spaces, depending on how the physical theories are formulated for these systems. The well-known mathematical categories are the category of groups, the category of vector spaces, the category of Lie algebras, the category of topological spaces, the Hilbert spaces and the Fock spaces~\cite{RG1,RS1}. 

In particular, quantum systems can be formulated within the quantum theory by means of some arbitrary wave functions, without directly taking into account the uniqueness of the chemical elements. On the other hand, the quantum systems that are formulated using the ionization energy theory belong to the renormalization group and quantum theories, and the formulations start from an entirely new notion, which is related to the uniqueness of energy-level spacings that exist in each chemical element, which has been proven to be valid in the earlier section. For classical systems however, one starts with some arbitrary continuous- and/or special-functions in accordance with classical physics. 

\section{Functors to compare $^{\xi}\mathscr{P}$}

To prove the existence of a mapping that properly maps one $^{\xi}\mathscr{P}$ to another $^{\xi}\mathscr{P}'$, one requires to invoke the notion of functors. This notion has been defined by Geroch [see page 90 in Ref.~\cite{RG1}], and is given in \ref{D24}.

\begin{definition}
\label{D24}
Let $\Omega_{\rm A}$, $\Omega_{\rm B}$, $\cdots$, $\Omega_{\rm Z}$ be the objects in category $\mathscr{C}$, while $\Gamma(\Omega_{\rm A})$, $\cdots$, $\Gamma(\Omega_{\rm Z})$ be the objects in category $\mathscr{C}'$. Given these background, a covariant functor $\Gamma$ needs two things--- $(i)$ a rule to associate each object in $\mathscr{C}$ to an object in $\mathscr{C}'$, and $(ii)$ a rule to associate each morphism, $\Phi_{\Omega_{\rm A},\Omega_{\rm B}}, \cdots$, in category $\mathscr{C}$ to a morphism $\Gamma(\Phi_{\Omega_{\rm A},\Omega_{\rm B}}), \cdots$, in category $\mathscr{C}'$. The rule in $(ii)$ needs to satisfy two other conditions listed below. \\
$(a)$ Composition is preserved--- Suppose $\Omega_{\rm A} \stackrel{\Phi_{\Omega_{\rm A},\Omega_{\rm B}}}{---\longrightarrow} \Omega_{\rm B} \stackrel{\Phi_{\Omega_{\rm B},\Omega_{\rm C}}}{---\longrightarrow} \Omega_{\rm C}$ is a diagram in $\mathscr{C}$, then there is a corresponding diagram in $\mathscr{C}'$ such that $\Gamma(\Omega_{\rm A}) \stackrel{\Gamma(\Phi_{\Omega_{\rm A},\Omega_{\rm B}})}{---\longrightarrow} \Gamma(\Omega_{\rm B}) \stackrel{\Gamma(\Phi_{\Omega_{\rm B},\Omega_{\rm C}})}{---\longrightarrow} \Gamma(\Omega_{\rm C})$ and therefore, the composition in $\mathscr{C}'$ is given by 
\begin{eqnarray}
&&\Gamma\big[\Phi_{\Omega_{\rm B}, \Omega_{\rm C}}\circ \Phi_{\Omega_{\rm A}, \Omega_{\rm B}}\big] = \Gamma(\Phi_{\Omega_{\rm B}, \Omega_{\rm C}})\circ \Gamma(\Phi_{\Omega_{\rm A}, \Omega_{\rm B}}). \label{eq:37}
\end{eqnarray}
$(b)$ Identities are preserved---For any object, $\Omega_{\rm A}$ in $\mathscr{C}$, there is a corresponding object $\Gamma(\Omega_{\rm A})$ in $\mathscr{C}'$ due to rule $(i)$, and therefore we have   
\begin{eqnarray}
&&\Gamma(\Phi^{\rm identity}_{\Omega_{\rm A},\Omega_{\rm A}}) = \Phi^{\rm identity}_{\Gamma\Omega_{\rm A},\Gamma\Omega_{\rm A}}. \label{eq:38}
\end{eqnarray} 
A contravariant functor $\Gamma$ needs two things--- $(i)$ a rule to associate each object in $\mathscr{C}$ to an object in $\mathscr{C}'$, and $(ii)$ a rule to associate each morphism, $\Phi_{\Omega_{\rm A},\Omega_{\rm B}}$, $\cdots$, in category $\mathscr{C}$ to a morphism $\Gamma(\Phi_{\Omega_{\rm B},\Omega_{\rm A}})$, $\cdots$, in category $\mathscr{C}'$. The rule in $(ii)$ also needs to satisfy two other conditions listed below. \\
$(c)$ Composition is preserved--- Suppose $\Omega_{\rm A} \stackrel{\Phi_{\Omega_{\rm A},\Omega_{\rm B}}}{---\longrightarrow} \Omega_{\rm B} \stackrel{\Phi_{\Omega_{\rm B},\Omega_{\rm C}}}{---\longrightarrow} \Omega_{\rm C}$ is a diagram in $\mathscr{C}$, then there is a corresponding diagram in $\mathscr{C}'$ such that $\Gamma(\Omega_{\rm C}) \stackrel{\Gamma(\Phi_{\Omega_{\rm C},\Omega_{\rm B}})}{---\longrightarrow} \Gamma(\Omega_{\rm B}) \stackrel{\Gamma(\Phi_{\Omega_{\rm B},\Omega_{\rm A}})}{---\longrightarrow} \Gamma(\Omega_{\rm A})$ and therefore, the composition in $\mathscr{C}'$ is given by 
\begin{eqnarray}
&&\Gamma\big[\Phi_{\Omega_{\rm B}, \Omega_{\rm C}}\circ \Phi_{\Omega_{\rm A}, \Omega_{\rm B}}\big] = \Gamma(\Phi_{\Omega_{\rm B}, \Omega_{\rm A}})\circ \Gamma(\Phi_{\Omega_{\rm C}, \Omega_{\rm B}}). \label{eq:39}
\end{eqnarray} 
$(d)$ Identities are preserved--- See $(b)$ above $\blacksquare$
\end{definition}

\begin{proposition}
\label{P25}
Neither a proper covariant nor a contravariant functor exist between any two $^{\xi}\mathscr{P}$.
\end{proposition}
\begin{proof}
It is sufficient to compare $(i)$ $^{\xi}\mathscr{P}^{\mathbb{Z}}_{\chi}$ with $^{\xi}\mathscr{P}^{\mathbb{Z}}_{\gamma}$ and $(ii)$ $^{\xi}\mathscr{P}^{\mathbb{Z}}_{\gamma}$ with $^{\xi}\mathscr{P}^{\mathbb{R}}_{\beta}$. $(i)$ From Part I in \ref{P21}, $\{\chi_{i,j}^{\alpha'-}\}, \cdots, \{\chi_{i,j}^{\alpha-}\}, \{\chi_{i,j}\}, \cdots, \{\chi_{i,j}^{\omega'+}\} \in$ $\mathcal{C}_{\chi_{i,j}} \in$ $^{\xi}\mathscr{P}^{\mathbb{Z}}_{\chi}$, and from Part II in \ref{P21}, $\{\gamma_j^{\alpha'-}\}$, $\cdots$, $\{\gamma_j\}$, $\cdots$, $\{\gamma_j^{\omega'+}\} \in$ $\mathcal{C}_{\gamma_j} \in$ $^{\xi}\mathscr{P}^{\mathbb{Z}}_{\gamma}$. Let $\gamma$ = H$_2$O, which implies there are only two elements ($\chi_i$ = H and $\chi_j$ = O) that can be related between $\mathcal{C}_{\chi_{i,j}}$ and $\mathcal{C}_{\gamma_j}$ such that $\xi_{\rm H^+_2O^{2-}} \propto 2\xi_{\rm H^+} + \xi_{\rm O^+}$. The proportionality here follows the proof given in \ref{P6}$(a)$ (after \ref{D15}) and the proof given in \ref{P6}$(b)$. $(ii)$ From Part III in \ref{P21}, $\mathcal{C}_{\beta_j} \in$ $^{\xi}\mathscr{P}^{\mathbb{R}}_{\beta}$, and if we let $\beta$ = Si$_{x_j}$, then for $x_i \gg x_j$, Si atoms may form a solid and therefore, there is not a single relation between $\mathcal{C}_{\beta_j}$ and $\mathcal{C}_{\gamma_j}$ because Si$_{x_i \gg x_j}$ (solids) do not exist in $\mathcal{C}_{\gamma_j}$. On the other hand, for $x_i \ll x_j$, Si atoms may be molecule-like. In this case, there is a one-to-one relation between $\mathcal{C}_{\beta_j}$ and $\mathcal{C}_{\gamma_j}$ such that $\xi_{{\rm Si}_{x_i \ll x_j}} = \xi_{{\rm Si}_{x_i \ll x_j}}$. Consequently, we cannot build any proper functor between $\mathcal{C}_{\beta_j}$ and $\mathcal{C}_{\gamma_j}$, and therefore, a proper functor (covariant or contravariant) does not exist between any two $^{\xi}\mathscr{P}$.      
\end{proof}
The proposition given in \ref{P25} implies the existence of different classes of quantum matter, namely, $(i)$ atoms and ions, $(ii)$ molecules and molecular ions, and $(iii)$ compounds such that there can be many categories within molecules and compounds (including nanoparticles), which can be constructed from the class of atoms and ions. In particular, for molecules we have the category of H$_2$O$_2$, the category of H$_2$SO$_4$, and so on. While for the compounds, we have the categories for free electron metals, the categories for superconductors, the categories for ferromagnets, and so on.  

\section*{Acknowledgments}

This work was supported by Sebastiammal Innasimuthu, Arulsamy Innasimuthu, Amelia Das Anthony, Malcolm Anandraj and Kingston Kisshenraj. Special thanks to Mir Massoud Aghili Yajadda (CSIRO, Lindfield) and Alexander Jeffrey Hinde (The University of Sydney) for providing some of the listed references.


\begin{thebibliography}{}

\bibitem{sch} E. Schr$\ddot{\rm o}$dinger, Ann. Phys. (Berlin) \textbf{379}, 361 (1926).

\bibitem{sch2} E. Schr$\ddot{\rm o}$dinger, Ann. Phys. (Berlin) \textbf{379}, 489 (1926).

\bibitem{sch3} E. Schr$\ddot{\rm o}$dinger, Phys. Rev. \textbf{28}, 1049 (1926).

\bibitem{planck}  M. Planck, Ann. Phys. (Berlin) \textbf{309}, 553 (1901).

\bibitem{heisen} W. Heisenberg, Zeitschrift f$\ddot{\rm u}$r Physik \textbf{43}, 172 (1927).

\bibitem{epr} A. Einstein, B. Podolski and N. Rosen, Phys. Rev. \textbf{47}, 777 (1935).

\bibitem{born} M. Born and P. Jordan, Zeitschrift f$\ddot{\rm u}$r Physik, \textbf{34}, 858 (1925) 

\bibitem{born2} M. Born, W. Heisenberg, and P. Jordan, Zeitschrift f$\ddot{\rm u}$r Physik, \textbf{35}, 557 (1925)

\bibitem{bethe} H. A. Bethe and E. E. Salpeter, \textit{Quantum mechanics of one- and two-electron atoms}, (Springer-Verlag, Berlin, 1957, Germany).

\bibitem{levin} I. N. Levin, \textit{Quantum chemistry}, (Prentice-Hall, New Jersey, 2000, USA).

\bibitem{par} R. G. Parr, W. Yang, \textit{Density functional theory of atoms and molecules}, (Oxford University Press, Oxford, 1989, UK).

\bibitem{pople} J. A. Pople and D. L. Beveridge, \textit{Approximate molecular orbital theory}, (McGraw-Hill, New York, 1970, USA).

\bibitem{kohn} W. Kohn and L. J. Sham, Phys. Rev. A \textbf{140}, 1133 (1965).

\bibitem{kohn2} P. Hohenberg and W. Kohn, Phys. Rev. B \textbf{136}, 864 (1964).

\bibitem{kohn3} W. Kohn, Rev. Mod. Phys. \textbf{71}, 1253 (1999).

\bibitem{pople2} J. A. Pople, Rev. Mod. Phys. \textbf{71}, 1267 (1999).

\bibitem{kaxi} E. Kaxiras \textit{Atomic and electronic structure of solids}, (Cambridge University Press, 2003, USA).

\bibitem{bo} M. Born and J. R. Oppenheimer, Ann. Phys. (Berlin) \textbf{389}, 457 (1927).

\bibitem{sla} J. C. Slater and H. C. Verma, Phys. Rev. \textbf{34}, 1293 (1929).

\bibitem{sla2} J. C. Slater, Phys. Rev. \textbf{81}, 385 (1951).

\bibitem{sla3} J. C. Slater, Phys. Rev. \textbf{91}, 528 (1953).

\bibitem{pauli} W. Pauli, Zeitschrift f$\ddot{\rm u}$r Physik, \textbf{31}, 765 (1925).

\bibitem{ps} P. Strange, \textit{Relativistic quantum mechanics with applications in condensed matter and atomic physics}, (Cambridge University Press, Cambridge, 1998, UK).

\bibitem{ADA1} A. D. Arulsamy, Pramana J. Phys. \textbf{74}, 615 (2010); Detailed derivations are given in A. D. Arulsamy, \textit{Many-body Hamiltonian based on ionization energy concept: a renormalized theory to study strongly correlated matter and nanostructures}, PhD thesis, The University of Sydney, Australia (2009).

\bibitem{ADA2} A. D. Arulsamy, Prog. Theor. Phys. \textbf{126}, 577 (2011).

\bibitem{ADA5} A. D. Arulsamy, Ann. Phys. (N.Y.) \textbf{326}, 541 (2011).

\bibitem{ADA4} A. D. Arulsamy, Phys. Lett. A \textbf{334}, 413 (2005).

\bibitem{gell} M. Gell-Mann and K. Brueckner, Phys. Rev. \textbf{106}, 364 (1957).

\bibitem{shank1} R. Shankar, Physica A \textbf{177}, 530 (1991).

\bibitem{shank2} R. Shankar, Rev. Mod. Phys. \textbf{66}, 129 (1994).

\bibitem{shank3} R. Shankar, Phil. Trans. R. Soc. A \textbf{369}, 2612 (2011).

\bibitem{ADA7} A. D. Arulsamy, Physica C \textbf{356}, 62 (2001).

\bibitem{ADA8} A. D. Arulsamy, Phys. Lett. A \textbf{300}, 691 (2002).

\bibitem{ja} C. S. Jayanthi, S. Y. Wu, J. Cocks, N. S. Luo, Z. L. Xie, M. Menon, G. Yang, Phys. Rev. B \textbf{57}, 3799 (1998).

\bibitem{ja2} S. Y. Wu, C. S. Jayanthi, Phys. Rep. \textbf{358}, 1 (2002).

\bibitem{ja3} M. Yu, I. Chaudhuri, C. Leahy, S. Y. Wu, C. S. Jayanthi, J. Chem. Phys. \textbf{130}, 184708 (2009).

\bibitem{ja4} M. Yu, S. Y. Wu, C. S. Jayanthi, Physica E \textbf{42}, 1 (2009).

\bibitem{DJG1} D. J. Griffiths, \textit{Introduction to quantum mechanics}, (Second Edition, Prentice-Hall, New Jersey, 2005, USA).

\bibitem{ADA3} A. D. Arulsamy, K. Eler$\check{\rm s}$i$\check{\rm c}$, M. Modic, U. Cvelbar and M. Mozeti$\check{\rm c}$, ChemPhysChem \textbf{11}, 3704 (2010).

\bibitem{AM1} N. W. Ashcroft and N. D. Mermin, \textit{Solid state physics}, (Holt, Rinehart and Winston, New York, 1976, USA)

\bibitem{ADA6} A. D. Arulsamy, arXiv:1107.4585 (2011).

\bibitem{AA1} A. Amira, M. F. Mosbah, A. Leblanc, P. Molinie and B. Corraze, Phys. Status Solidi C \textbf{1}, 1944 (2004).

\bibitem{ADA10} A. D. Arulsamy and M. Fronzi, Physica E \textbf{41}, 74 (2008).

\bibitem{ADA11} A. D. Arulsamy, X. Y. Cui, C. Stampfl and K. Ratnavelu, Phys. Status Solidi B \textbf{246}, 1060 (2009).

\bibitem{ADA12} A. D. Arulsamy and K. Ostrikov, J. Supercond. Nov. Magn. \textbf{22}, 785 (2009).

\bibitem{SYH1} S. Y. Huang, A. D. Arulsamy, M. Xu, S. Xu, U. Cvelbar, M. Mozeti$\check{\rm c}$ and K. Ostrikov, Phys. Plasmas \textbf{16}, 123504 (2009).

\bibitem{AER1} A. D. Arulsamy, A. E. Rider, Q. J. Cheng, S. Xu and K. Ostrikov, J. Appl. Phys. \textbf{105}, 094314 (2009).

\bibitem{ADA14} A. D. Arulsamy and K. Ostrikov, Phys. Lett. A \textbf{373}, 2267 (2009).

\bibitem{MM1} M. Mahtali, E. H. Boudjema, R. Labbani, S. Chamekh, A. Bouabellou, A. Toufik and C. Simon, Surf. Interf. Anal. \textbf{42}, 935 (2010).

\bibitem{DHS1} D. H. Seo, A. E. Rider, A. D. Arulsamy, I. Levchenko and K. Ostrikov, J. Appl. Phys. \textbf{107}, 024313 (2010).

\bibitem{ADA15} A. D. Arulsamy, U. Cvelbar, M. Mozeti$\check{\rm c}$ and K. Ostrikov, Nanoscale \textbf{2}, 728 (2010).

\bibitem{ADA16} A. D. Arulsamy and K. Ostrikov, Physica B \textbf{405}, 2263 (2010).

\bibitem{ADA17} A. D. Arulsamy, Z. Kregar, K. Eler$\check{\rm s}$i$\check{\rm c}$, M. Modic and U. S. Subramani, Phys. Chem. Chem. Phys. \textbf{13}, 15175 (2011).

\bibitem{FS1} F. Shi and H. Dong, Dalton Trans. \textbf{40}, 6659 (2011).

\bibitem{DR1} D. Ratchford, K. Dziatkowski, T. Hartsfield, X. Li, Y. Gao and Z. Tang, J. Appl. Phys. \textbf{109}, 103509 (2011).

\bibitem{KD1} K. Dziatkowski, D. Ratchford, T. Hartsfield, X. Li, Y. Gao and Z. Tang, Acta Phys. Polon. A \textbf{120}, 870 (2011).

\bibitem{KE1} K. Eler$\check{\rm s}$i$\check{\rm c}$, M. Pi$\check{\rm c}$man, N. Hauptman, U. Cvelbar and M. Mozeti$\check{\rm c}$, IEEE Trans. Plasma Sci. \textbf{39}, 2812 (2011).

\bibitem{ADA18} A. D. Arulsamy, arXiv:1105.5862 (2011).

\bibitem{ADA19} A. D. Arulsamy, arXiv:1110.3412 (2011).

\bibitem{ADA20} A. D. Arulsamy, arXiv:1109.1259 (2011).

\bibitem{radha} A. N. Radhakrishnan, P. P. Rao, S. K. Mahesh, D. S. V. Thampi and P. Koshy, Inorg. Chem. \textbf{51}, 2409 (2012).

\bibitem{RG1} R. Geroch, \textit{Mathematical physics}, (The University of Chicago Press, Chicago, 1985, USA).

\bibitem{RS1} R. Shankar, \textit{Principles of quantum mechanics}, (Second Edition, Springer, New York, 1994, USA).

\end{thebibliography}
\end{document}